\documentclass[letterpaper,11pt]{article}
\usepackage[margin=1in]{geometry}

\usepackage{amssymb}
\usepackage{setspace}
\usepackage{authblk}
\usepackage[title]{appendix}\usepackage{amsmath}%
\usepackage{amsfonts}%
\usepackage{amssymb}%
\usepackage{amsthm}
\usepackage{breqn}
\usepackage{graphicx}
\usepackage{colonequals}
\usepackage{psfrag}
\usepackage{subfigure}
\usepackage{mathrsfs}
\usepackage[noadjust]{cite}
\usepackage{hyperref}
\usepackage{bm,bbm}

\usepackage[usenames,dvipsnames]{xcolor}
\usepackage{mdframed}
\usepackage{bm}
\usepackage{bbm}

\hypersetup{%
  colorlinks=true,
  urlcolor=blue,%
  linkcolor=blue,%
  citecolor=blue,%
  linkbordercolor=blue,
  pdfborderstyle={/S/U/W 1}
}

\usepackage{prettyref,enumerate}
\usepackage[]{algorithm2e}

\newtheorem{thm}{Theorem}
\newtheorem{lem}{Lemma}

\newtheorem{cor}{Corollary}

\theoremstyle{definition}

\newtheorem{defn}{Definition}

\newtheorem{rem}{Remark}
\newtheorem{example}{Example}

\newcommand{\diag}{\mathsf{diag}}

\newcommand{\etal}{\textit{et al.}~}

\newcommand{\calX}{\mathcal{X}}
\newcommand{\calA}{\mathcal{A}}
\newcommand{\calB}{\mathcal{B}}

\newcommand{\calY}{\mathcal{Y}}

\newcommand{\calS}{\mathcal{S}}

\newcommand{\calL}{\mathcal{L}}

\newcommand{\bD}{\mathbf{D}}
\newcommand{\bx}{\mathbf{x}}

\renewcommand{\tilde}{\widetilde}

\DeclareMathOperator*{\argmin}{\arg\!\min}
\DeclareMathOperator*{\argmax}{\arg\!\max}

\newcommand{\Reals}{\mathbb{R}}

\newcommand{\defined}{\triangleq}

\newcommand{\ExpVal}[2]{\mathbb{E}\left[ #2 \right]}

\newcommand{\by}{\mathbf{y}}

\newcommand{\bnu}{\pmb{\nu}}

\newcommand{\bu}{\mathbf{u}}

\newcommand{\ba}{\mathbf{a}}
\newcommand{\bb}{\mathbf{b}}

\newcommand{\sto}{\mbox{\normalfont s.t.}}

\newcommand{\EE}[1]{\ExpVal{}{#1}}

\newcommand{\bP}{\mathbf{P}}

\newcommand{\bQ}{\mathbf{Q}}
\newcommand{\bU}{\mathbf{U}}
\newcommand{\bSigma}{\mathbf{\Sigma}}
\newcommand{\bV}{\mathbf{V}}

\newcommand{\brho}{\bm{\rho}}

\newcommand{\indicator}{\mathbbm{1}}
\newcommand{\mmse}{\mathsf{mmse}}

\newcommand{\tr}{\mathsf{tr}}
\newcommand{\bA}{\mathbf{A}}
\newcommand{\bB}{\mathbf{B}}
\newcommand{\bW}{\mathbf{W}}
\newcommand{\bM}{\mathbf{M}}
\newcommand{\bL}{\mathbf{L}}
\newcommand{\indep}{\rotatebox[origin=c]{90}{$\models$}}

\newcommand{\textfn}[1]{{\small\textit{#1}}}

\providecommand{\keywords}[1]
{
  \small	
  \textbf{\textit{Index Terms---}} #1
}

\title{Privacy with Estimation Guarantees}
\author{Hao Wang, Lisa Vo, Flavio P. Calmon, Muriel M\'edard, Ken R. Duffy, Mayank Varia
\thanks{This material is based upon work supported by the National Science Foundation under Grant No. CCF-1845852.\\
This paper was presented in part at the 2014 52nd Annual Allerton Conference on Communication, Control, and Computing \cite{calmon2014information}, and the 2017 55th Annual Allerton Conference on Communication, Control, and Computing \cite{wang2017estimation}.\\
H.~Wang, L.~Vo, and F.~P.~Calmon are with Harvard University (e-mail: hao\_wang@g.harvard.edu; lisavo@college.harvard.edu; flavio@seas.harvard.edu).\\
M.~M\'edard is with the Massachusetts Institute of Technology (e-mail: medard@mit.edu).
\\
K.~R.~Duffy is with the Hamilton Institute, Maynooth University (e-mail: ken.duffy@nuim.ie).
\\
M.~Varia is with Boston University (e-mail: varia@bu.edu).
}
}
\date{}

\begin{document}
\maketitle
\vspace{-.2in}

\begin{abstract}
We study the central problem in data privacy: how to share data with an analyst while providing both privacy and utility guarantees to the user that owns the data. In this setting, we present an estimation-theoretic analysis of the privacy-utility trade-off (PUT). Here, an analyst is allowed to reconstruct (in a mean-squared error sense) certain functions of the data (utility), while other private functions should not be reconstructed with distortion below a certain threshold (privacy). We demonstrate how chi-square information captures the fundamental PUT in this case and provide bounds for the best PUT. We propose a convex program to compute privacy-assuring mappings when the functions to be disclosed and hidden are known a priori and the data distribution is known. We derive lower bounds on the minimum mean-squared error of estimating a target function from the disclosed data and evaluate the robustness of our approach when an empirical distribution is used to compute the privacy-assuring mappings instead of the true data distribution. We illustrate the proposed approach through two numerical experiments.
\end{abstract}


\newpage
\section{Introduction}

Data sharing and publishing is increasingly common within scientific communities \cite{tenopir2011data}, businesses  \cite{stefansson2002business}, government operations \cite{otjacques2007interoperability}, medical fields \cite{sweeney1997guaranteeing}, and beyond. Data is usually shared with an application in mind, from which the data provider receives some utility. For example, when a user shares her movie ratings with a streaming service, she receives utility in the form of suggestions of new, interesting movie recommendations that fit her taste. As a second example, when a medical research group shares patient data, their aim is to enable a wider community of researchers and statisticians to learn patterns from that data. Utility is then gained through new scientific discoveries.

The disclosure of non-encrypted data incurs a privacy risk through unwanted inferences. In our previous examples, the streaming service may infer the user's political preference (potentially deemed private by the user) from her movie ratings~\cite{narayanan2008robust}, or an insurance company may determine the identity of a patient within a medical dataset \cite{sweeney1997guaranteeing,sweeney2015only,safran2007toward}. If privacy is a concern but the data has no immediate utility, then cryptographic methods suffice.

The dichotomy between privacy and utility has been widely studied by computer scientists, statisticians, and information theorists alike. While specific metrics and  models vary among these communities, their desideratum is the same: to design mechanisms that perturb the data (or functions thereof) while achieving an acceptable privacy-utility trade-off (PUT). The feasibility of this goal depends on several factors including the chosen privacy and utility metric, as well as the topology and distribution of the data. The information-theoretic approach to privacy, and notably the results by Sankar \etal \cite{sankar2013utility,sankar2010theory}, Issa \etal \cite{issa2016operational,issa2016maximal}, Asoodeh \etal \cite{asoodeh2016information,asoodeh2016privacy}, Calmon \etal \cite{calmon2017principal,du2012privacy}, among others, seek to quantify the best possible PUT for \textit{any} privacy mechanism. In those works, information-theoretic quantities, such as mutual information and maximal leakage \cite{issa2016operational,issa2016maximal}, have been used to characterize privacy, and bounds on the fundamental PUT were derived under assumptions on the distribution of the data. It is within this information-theoretic approach that the present work is inscribed.

Our aim is to characterize the fundamental limits of PUT from an estimation-theoretic perspective, and to design privacy-assuring mechanisms that provide estimation-theoretic guarantees. 
We use the \textit{principal inertia components} (PICs) \cite{calmon2017principal,hirschfeld1935connection,gebelein1941statistische,sarmanov1962maximum,witsenhausen1975sequences,greenacre1987geometric,renyi1959measures,buja1990remarks,makur2015efficient,makur2017polynomial} to formalize the privacy and utility constraints. 
The PICs quantify the minimum mean-squared error (MMSE) achievable for reconstructing both private and useful information from the disclosed data. We do not seek to claim that the estimation-based approach subsumes other privacy metrics, such as differential privacy \cite{dwork2011differential}. Rather, our goal is to show that the MMSE viewpoint reveals an interesting facet of data disclosure which, in turn, can drive the design of privacy mechanisms used in practice.

In the remainder of this section, we present an overview of the paper and our main results, discuss related work, and introduce the notation adopted in the paper.

\subsection{Overview and Main Contributions}
\label{sec:BigPicture}

Throughout this paper, we assume all random variables are discrete with finite support sets. We let $S$ denote a private variable to be hidden (e.g., political preference) and $X$ be a useful variable that depends on $S$ (e.g., movie ratings). Our goal is to disclose a realization of a random variable $Y$, produced from $X$ through a randomized mapping $P_{Y|X}$ called the \textit{privacy-assuring mapping}. Here, $S$, $X$, and $Y$ satisfy the Markov condition $S\to X\to Y$. We assume that an analyst will provide some utility based on an observation of $Y$ (e.g., movie recommendations), while potentially trying to estimate $S$ from $Y$.
Denoting $[n]\defined \left\{1,\dots,n \right\}$, the support sets of $S$, $X$, and $Y$ are $\mathcal{S} =\left[|\mathcal{S}|\right]$, $\mathcal{X}=\left[|\mathcal{X}|\right]$, and $\mathcal{Y}=\left[|\mathcal{Y}|\right]$, respectively.

In the sequel, we derive PUTs when both privacy and utility are measured in terms of the mean-squared error of reconstructing functions of $S$ and $X$ from an observation of $Y$. We analyze three related scenarios: (i) an \textit{aggregate} setting, where certain functions of $X$ can be, on average, reconstructed from the disclosed variable while controlling the MMSE of estimating functions of $S$ and $P_{S,X}$ is known to the privacy mechanism designer, (ii) a \textit{composite} setting, where specific functions of $S$ and $X$ have different privacy/utility reconstruction requirements and $P_{S,X}$ is known to the privacy mechanism designer, and (iii) a \textit{restricted-knowledge} setting, where $P_{S,X}$ is unknown, but the correlation between a target function to be hidden and a set of functions which are known to be hard to infer from the disclosed variable is given. For the first two thrusts, we also analyze the robustness of privacy-assuring mappings designed using an empirical estimate of $P_{S,X}$ computed from a finite number of samples. Next, we present the outline of the paper and a summary of our main contributions.

\subsubsection*{Aggregate PUTs}

We start by studying the problem of limiting an untrusted party's ability to estimate functions of $S$ given an observation of $Y$, while controlling for the MMSE of reconstructing functions of $X$ given $Y$. Here, privacy and utility are measured in terms of the $\chi^2$-information between $S$ and $Y$ and the $\chi^2$-information between $X$ and $Y$, denoted by $\chi^2(S;Y)$ and $\chi^2(X;Y)$ (cf. \eqref{eq:def_chi_sq_inf}), respectively.
We introduce the $\chi^2$-privacy-utility function in Section~\ref{sec:u-p tradeoff}. Bounds of this function are presented in Theorem~\ref{thm:bound}. In particular, the upper bound is cast in terms of the PICs of $P_{S,X}$ and provides an interpretation of the trade-off between privacy and utility that goes beyond simply using maximal correlation. We also prove that the upper bound is achievable in the high-privacy regime in Theorem~\ref{thm:highprivacyregion}.

\subsubsection*{Composite PUTs}
$\chi^2$-based metrics guarantee privacy and utility in a uniform sense, capturing the aggregate mean-squared error of estimating \textit{any} functions of the private and the useful variables. However, in many applications, specific functions of $S$ and $X$ that should be hidden/revealed are known \textit{a priori}. This knowledge enables a more refined design of privacy-assuring mechanisms that  specifically target these functions. We explore this finer-grained approach in Section~\ref{sec:add_opt}, and propose a PIC-based convex program for computing privacy-assuring mappings within this setting. We demonstrate the practical feasibility of the convex programs through two numerical experiments in Section~\ref{sec:numericalresults}, deriving privacy-assuring mappings for a synthetic dataset and a real-world dataset. In the latter case, we approximate $P_{S,X}$ using its empirical distribution.

\subsubsection*{Restricted Knowledge of the Distribution}
The aforementioned aggregate and composite PUTs require knowledge of the joint distribution $P_{S,X}.$ In Section~\ref{sec:MMSE}, we forgo this assumption, and study  a simpler setting where $S=\phi(X)$ (i.e., the private variable is a function of the data) and the correlation between $\phi(X)$ and a set of functions (composed with the data) $\left\{\phi_j(X)\right\}_{j=1}^m$ is given. In practice, $\phi(X)$ may be a sensitive feature of the data $X$, and $\left\{\phi_j(X)\right\}_{j=1}^m$ is a collection of other features from which $\EE{\phi(X)\phi_j(X)}$ can be accurately estimated. 

Our goal here is to derive lower bounds on the MMSE of estimating a real-valued function of X, namely $\phi(X)$, from $Y$ for any privacy-assuring mapping $P_{Y|X}$. These bounds are cast in terms of the MMSE of estimating $\phi_j(X)$ from $Y$ and the correlation between $\phi(X)$ and $\left\{\phi_j(X)\right\}_{j=1}^m$. This leads to a converse result in Theorem~\ref{thm:tighter}: if the MMSE of estimating $\phi_j(X)$ from $Y$ is large and $\phi(X)$ is strongly correlated with $\phi_j(X)$, then the MMSE of estimating $\phi(X)$ from $Y$ will also be large and privacy is assured in an estimation-theoretic sense. The inverse result is straightforward: if $\phi(X)$ and $\phi_j(X)$ are strongly correlated and $\phi_j(X)$ can be reliably reconstructed from $Y$, then $\phi(X)$ can also be reliably estimated from $Y$. This intuitive trade-off is at the heart of the estimation-theoretic view of privacy, and demonstrates that no function of $X$ can remain private whilst other strongly correlated functions are revealed through $Y$. The results in Section~\ref{sec:MMSE} make this intuition mathematically precise.

Finally, in Section~\ref{sec:robust} we investigate the resilience of  privacy-assuring mappings when designed using an estimate of the distribution $P_{\hat{S},\hat{X}}$ computed as the empirical frequencies of $S,X$ obtained from $n$ i.i.d. samples. Here, the value of the privacy and utility guarantees estimated using $P_{\hat{S},\hat{X}}$ will not match the true values $\chi^2(S;Y)$ and $\chi^2(X;Y)$ obtained when the privacy-assuring mechanism is applied to fresh samples drawn from the true distribution $P_{S,X}$. We bound this performance gap in Theorem~\ref{thm:robust} and show that this gap scales as $O\left(\sqrt{1/n}\right)$, while also depending on the alphabet size of the variables and the probability of the least likely symbols.

\subsection{Related Work}

Currently, the most adopted definition of privacy is differential privacy \cite{dwork2011differential,dwork2006calibrating}, which enables queries to be computed over a database while simultaneously ensuring privacy of individual entries of the database. Information-theoretic quantities, such as R{\'e}nyi divergence, can be used to relax the definition of differential privacy \cite{mironov2017renyi}. Fundamental bounds on composition of differentially private mechanisms were given by Kairouz \etal \cite{kairouz2017composition}. Recently, a new privacy framework called Pufferfish \cite{kifer2014pufferfish} was developed for creating customized privacy definitions. 

Several papers, such as Sankar \etal \cite{sankar2013utility}, Calmon and Fawaz \cite{du2012privacy}, Asoodeh \etal \cite{asoodeh2014notes}, and Makhdoumi \etal \cite{makhdoumi2014information}, have studied information disclosure with privacy guarantees through an information-theoretic lens. For example, Sankar \etal \cite{sankar2013utility} characterized PUTs in large databases using tools from rate-distortion theory. Calmon and Fawaz \cite{du2012privacy} used expected distortion and mutual information to measure utility and privacy, respectively, and characterized the PUT as an optimization problem. Makhdoumi \etal \cite{makhdoumi2014information} introduced the privacy funnel, where both privacy and utility are measured in terms of mutual information, and showed its connection with the information bottleneck \cite{tishby2000information}. The PUT was also explored in  \cite{makhdoumi2013privacy} and \cite{rebollo2010t} using mutual information as a privacy metric. 

Other quantities from the information-theoretic literature have been used to quantify privacy and utility. For example, Asoodeh \etal \cite{asoodeh2016information} and Calmon \etal \cite{calmon2017principal} used estimation-theoretic tools to characterize fundamental limits of privacy. Liao \etal \cite{liao2016hypothesis,liao2017hypothesis} explored the PUT within a hypothesis testing framework. Issa \etal \cite{issa2016operational,issa2017operational} introduced maximal leakage as an information leakage metric. There is also significant recent work in information-theoretic privacy in the context of network secrecy. For example, Li and Oechtering \cite{li2015privacy} proposed a new privacy metric based on distributed Bayesian detection which can inform privacy-aware system design. Recently, Tripathy \etal \cite{tripathy2017privacy} and Huang \etal \cite{huang2017context} used adversarial networks for designing privacy-assuring mappings that navigate the PUT. Takbiri \etal \cite{takbiri2017matching} considered obfuscation and anonymization techniques and characterized the conditions required to obtain perfect privacy.

MMSE-based analysis and maximal correlation have been investigated in the context of log-Sobolev inequalities and hypercontractivity, such as in the work of Raginsky \cite{raginsky2016strong}, Anantharam \etal \cite{anantharam2013hypercontractivity}, and Polyanskiy and Wu \cite{polyanskiy2016dissipation}. The metric used in this paper, namely $\chi^2$-information, relates with $\chi^2$-divergence, which is a special case of $f$-divergence \cite{liese2006divergences}. Also of note, the study of robustness of estimated distributions with finite sample size has appeared in \cite{paninski2003estimation,shamir2008learning,calmon2017optimized,wang2018utility}. 

\subsection{Notation}
\label{sec:notation}
Matrices are denoted in bold capital letters (e.g., $\bP$) and vectors in bold lower-case letters (e.g., $\mathbf{p}$). For a vector $\mathbf{p}$, $\diag(\mathbf{p})$ is defined as the matrix with diagonal entries equal to $\mathbf{p}$ and all other entries equal to $0$.
The span of a set $\mathcal{V}$ of vectors is
\begin{equation*}
\mathsf{span}(\mathcal{V}) \defined \left\{\sum_{i=1}^k \lambda_i \mathbf{v}_i\ \Big|\ k\in \mathbb{N}, \mathbf{v}_i\in \mathcal{V}, \lambda_i \in \mathbb{R}\right\}.
\end{equation*}
The dimension of a linear span is denoted by $\mathsf{dim}(\mathsf{span}(\mathcal{V}))$.

We denote independence of random variables $U$ and $V$ by $U \indep V$, and write $U\sim V$ to indicate that $U$ and $V$ have the same distribution. When $U$, $V$, and $W$ form a Markov chain, we write $U \rightarrow V \rightarrow W$. For a random variable $U$ with probability distribution $P_U$, we denote 
\begin{align*}
    P_{U\min} \defined \inf\{P_{U}(u)\mid u\in \mathcal{U}\},
\end{align*}
where $\mathcal{U}$ is the support set of $U$. The MMSE of estimating $U$ given $V$ is
\begin{equation*}
\begin{aligned}
\mmse(U|V) 
\defined \min_{U\rightarrow V \rightarrow \hat{U}} \EE{(U-\hat{U})^2} 
= \EE{(U-\EE{U|V})^2}.
\end{aligned}
\end{equation*}
The \textit{$\chi^2$-information} between two random variables $U$ and $V$ is defined as
\begin{equation}
\label{eq:def_chi_sq_inf}
\chi^2(U;V) \defined \EE{\left(\frac{P_{U,V}(U,V)}{P_U(U)P_V(V)}\right)}-1.
\end{equation}
Let $P_U$ and $Q_U$ be two probability distributions taking values in the same discrete and finite set $\mathcal{U}$. We denote $||P_U-Q_U||_1\defined \sum_{u\in \mathcal{U}}|P_U(u)-Q_U(u)|$. For any real-valued random variable $U$, we denote the $\calL_p$-norm of $U$ as
$$||U||_p \defined (\EE{|U|^p})^{1/p}.$$
The set of all functions that applied to a random variable $U$ with distribution $P_U$ result in an $\calL_2$-norm less than or equal to 1 is given by
\begin{equation}
    \calL_2(P_U)\defined \left\{f:\mathcal{U}\to \Reals\mid \|f(U)\|_2\leq1
  \right\}.
\end{equation}
The conditional expectation operators $T_{V|U}: \calL_2(P_V) \to \calL_2(P_U)$ and $T_{U|V}: \calL_2(P_U) \to \calL_2(P_V)$ are given by $(T_{V|U}g)(u)\defined \EE{g(V)|U=u}$ and $(T_{U|V}f)(v)\defined \EE{f(U)|V=v}$, respectively.
\section{Principal Inertia Components}
\label{sec:PIC}

We present next the properties of the PICs that will be used in this paper. For a more detailed overview, we refer the reader to \cite{calmon2017principal} and the references therein. We use the definition of PICs presented in \cite{calmon2017principal}, but note that the PICs predate \cite{calmon2017principal} by many decades (e.g., \cite{hirschfeld1935connection,gebelein1941statistische,sarmanov1962maximum,witsenhausen1975sequences,greenacre1987geometric,renyi1959measures,buja1990remarks}). 
Recently, Huang~\etal \cite{huang2014efficient} considered the PICs by analyzing the ``divergence transition matrix'' \cite[Eq.~2]{huang2014efficient}. Specifically, there are different directions of local perturbation \cite{borade2008euclidean} of input distribution and the direction which leads to the greatest influence of the output distribution of a noisy channel can be identified \cite{huang2014efficient} by specifying the singular vector decomposition of the divergence transition matrix. In follow-on work, Huang~\etal \cite{huang2017information} used the divergence transition matrix in the context of feature selection. The singular values of the divergence transition matrix are exactly the square root of the PICs considered here, and are also related to the singular values of the conditional expectation operator, as also noted by Makur and Zheng \cite{makur2017polynomial} and originally by Witsenhausen \cite{witsenhausen1975sequences} and others \cite{buja1990remarks}. We build on these prior works by using the PICs for quantifying privacy-utility trade-offs.

\begin{defn}[{\cite[Definition~1]{calmon2017principal}}]
\label{def:PIC}
Let $U$ and $V$ be random variables with support sets $\mathcal{U}$ and $\mathcal{V}$, respectively, and joint distribution $P_{U,V}$. In addition, let $f_0:\mathcal{U}\to \Reals$ and $g_0:\mathcal{V}\to \Reals$ be the constant functions $f_0(u)=1$ and $g_0(v)=1$. For $k\in\mathbb{Z}_+$, we (recursively) define
\begin{align} 
\label{eq::defn_pic_cor}
\lambda_k(U;V) 
\defined \EE{f_k(U)g_k(V)}^2,    
\end{align}
where
\begin{equation}
\label{eq:fkmaxcorr} 
\begin{aligned}
(f_k,g_k) \defined \argmax\Big\{\EE{f(U)g(V)}^2\ \Big|\ 
&f\in \calL_2(P_U),g\in\calL_2(P_V),\EE{f(U)f_j(U)}=0, \\
&\EE{g(V)g_j(V)}=0,j\in\{0,\dots,k-1\} \Big\}. 
\end{aligned}
\end{equation}
The values $\lambda_k(U;V)$ are called the \textit{principal inertia components} (PICs) of $P_{U,V}$. The functions $f_k$ and $g_k$ are called the \textit{principal functions} of $P_{U,V}$.
\end{defn}

Observe that the PICs satisfy $\lambda_k(U;V)\leq 1$, since $f_k\in \calL_2(P_U)$, $g_k\in \calL_2(P_V)$, and 
\begin{equation*}
  \left|\EE{f(U)g(V)}\right|\leq \|f(U)\|_2\|g(V)\|_2\leq1.
\end{equation*}
Thus, from Definition~\ref{def:PIC}, $0\leq \lambda_{k+1}(U;V)\leq \lambda_{k}(U;V)\leq 1$. 

The largest PIC satisfies $\lambda_1(U;V) = \rho_m(U;V)^2$ where $\rho_m(U;V)$ is the maximal correlation \cite{renyi1959measures}, defined as
\begin{align}
\label{eq:def_MC}
    \rho_m(U;V) \defined \max_{\substack{\EE{f(U)}=\EE{g(V)}=0\\\EE{f(U)^2}=\EE{g(V)^2}=1}} \EE{f(U)g(V)}.
\end{align}

\begin{defn}[{\cite[Definition~2]{calmon2017principal}}]
\label{def:Q}
  For $\mathcal{U}=[m]$ and $\mathcal{V}=[n]$, let $\bP_{U,V}\in \Reals^{m\times n}$ be a matrix with
  entries $[\bP_{U,V}]_{i,j}=P_{U,V}(i,j)$, and $\bD_U\in \Reals^{m\times m}$ and
  $\bD_V\in \Reals^{n\times n}$ be diagonal matrices with diagonal entries
  $[\bD_U]_{i,i}=P_U(i)$ and $[\bD_V]_{j,j}=P_V(j)$, respectively, where $i\in [m]$ and $j\in [n]$. We define
  \begin{equation}
    \label{eq:Qdefn}
    \bQ_{U,V} \defined \bD_U^{-1/2}\bP_{U,V}\bD_V^{-1/2}.
  \end{equation}
  We denote the singular value decomposition of $\bQ_{U,V}$ by $\bQ_{U,V}=\bU\bSigma\bV^T$.
\end{defn}

\begin{defn}[{\cite[Definition~14]{calmon2017principal}}]
\label{defn:lambda_min}
Let $d\defined \min\{|\mathcal{U}|,|\mathcal{V}|\}-1$, and $\lambda_d(U;V)$ the $d$-th PIC of $P_{U,V}$. We define
  \begin{equation}
  \label{equa:delta}
    \delta(P_{U,V}) \defined 
        \begin{cases}
            \lambda_d(U;V)& \text{if}\ |\mathcal{V}|\leq |\mathcal{U}|,\\
             0 & \text{otherwise.}
        \end{cases}
  \end{equation}
We also denote $\lambda_d(U;V)$ and the corresponding principal functions $f_d$, $g_d$ as $\lambda_{\min}(U;V)$ and $f_{\min}$, $g_{\min}$, respectively, when the alphabet size is clear from the context.
\end{defn}

The next theorem illustrates the different characterizations of the PICs used in this paper.

\begin{thm}[{\cite[Theorem~1]{calmon2017principal}}]
\label{thm:PIC_Charac}
The following characterizations of the PICs are equivalent:
\begin{enumerate}
\item The characterization given in Definition~\ref{def:PIC}, where, for $f_k$ and $g_k$ given in \eqref{eq:fkmaxcorr}, $g_k(V)=\frac{\EE{f_k(U)|V}}{\|\EE{f_k(U)|V}\|_2}$ and $f_k(U)=\frac{\EE{g_k(V)|U}}{\|\EE{g_k(V)|U}\|_2}$. 
\item For any $k\in\mathbb{Z}_+$, 
\begin{align}
\label{eq:mmsePIC}
1- \lambda_k(U;V) = \mmse(h_k(U)|V),    
\end{align}
where
\begin{align}
h_k \defined \argmin\Big\{ \mmse(h(U)|V)\ \Big|\ 
\|h(U)\|_2=1,\EE{h(U)h_j(U)}=0,j\in\{0,\dots,k-1\} \Big\}.\label{eq:fkmmse}    
\end{align}
If $\lambda_k(U;V)$ is unique, then $h_k=f_k$, given in \eqref{eq:fkmaxcorr}.
\item $\sqrt{\lambda_k(U;V)}$ is the $(k+1)$-st largest singular value of $\bQ_{U,V}$. The principal functions $f_k$ and $g_k$ in \eqref{eq:fkmaxcorr} correspond to the columns of the matrices $\bD_U^{-1/2}\bU$ and $\bD_V^{-1/2}\bV$, respectively, where $\bQ_{U,V} = \bU\bSigma\bV^T$. 
\end{enumerate}
\end{thm}

The equivalent characterizations of the PICs in the  above theorem have the following intuitive interpretation: the principal functions can be viewed as a basis that decompose the mean-squared error of estimating functions of a hidden variable $U$ given an observation $V$. In particular, for any zero-mean finite-variance function $f:\mathcal{U}\to \Reals$,
\begin{align*}
    \mmse(f(U)|V)
    = \sum_{i=1}^{|\mathcal{U}|-1} \EE{f(U)f_i(U)}^2(1-\lambda_i(U;V)).
\end{align*}

We remark that the $\chi^2$-information between $U$ and $V$ is the sum of all PICs. Specifically, it has been shown (e.g., \cite{calmon2017principal,witsenhausen1975sequences}) that $\chi^2(U;V) = \sum_{i=1}^d \lambda_i(U;V)$, where $d=\min \{|\mathcal{U}|, |\mathcal{V}|\}-1$.

\section{Aggregate PUTs: \\
The Chi-Square-Privacy-Utility Function}
\label{sec:u-p tradeoff}

We start our analysis by adopting $\chi^2$-information as a measure of both privacy and utility. As seen in the previous section,  $\chi^2(S;Y) = \sum_{i=1}^d \lambda_i(S;Y)$, where $d=\min \{|\mathcal{S}|, |\mathcal{Y}|\}-1$. 
If $\chi^2(S;Y)<1$, then, from characterization 2 in Theorem~\ref{thm:PIC_Charac}, the MMSE of reconstructing \textit{any} zero-mean, unit-variance function  of $S$ given $Y$ is lower bounded by $1-\chi^2(S;Y)$, i.e., all functions of $S$ cannot be reconstructed with small MMSE given an observation of $Y$.
Note that this argument also holds true when we replace $\chi^2$-information with the maximal correlation.
In fact, in the high privacy regime, the PUT under $\chi^2$-information is essentially equivalent to the PUT when both privacy and utility are measured using maximal correlation. We make this intuition precise at the end of this section.
When $1 \leq \chi^2(S;Y)$, certain private functions, on average, may be estimated from $Y$ but, in general, most private functions are still kept in secret.
Analogously, when $\chi^2(X;Y)$ is large, certain functions of $X$ can be, on average, reconstructed (i.e., estimated) with small MMSE from $Y$.
We demonstrate next that the PICs play a central role in bounding the PUT in this regime.

We first introduce the $\chi^2$-privacy-utility function. This function captures how well an analyst can reconstruct functions of the useful variable $X$ while restricting the analyst's ability to estimate functions of the private variable $S$.
\begin{defn}
\label{def:utility-privacy}
For a given joint distribution $P_{S,X}$ and $0\leq \epsilon \leq \chi^2(S;X)$, we define the $\chi^2$-privacy-utility (trade-off) function as
\begin{align*}
    F_{\chi^2}(\epsilon;P_{S,X}) \defined \sup_{P_{Y|X}\in \mathcal{D}(\epsilon;P_{S,X})}\chi^2(X;Y),
\end{align*}
where $\mathcal{D}(\epsilon;P_{S,X}) \defined \{ P_{Y|X} \mid S\rightarrow X\rightarrow Y, \chi^2(S;Y)\leq \epsilon\}.$
\end{defn}

It has been proved in \cite{witsenhausen1975conditional,hsu2018generalizing} that there is always a privacy-assuring mapping $P_{Y|X}$ which achieves the supremum in $F_{\chi^2}(\epsilon;P_{S,X})$ using at most $|\mathcal{X}|+1$ symbols (i.e., $|\mathcal{Y}|\leq |\mathcal{X}|+1$). 
The following lemma gives an alternative way to compute the $\chi^2$-information, in the discrete, finite setting.
\begin{lem}
\label{lem:trace}
Suppose $S\rightarrow X \rightarrow Y$. Then
\begin{equation}
\chi^2(X;Y) = \tr(\bA) -1,
\end{equation}
\begin{equation}
\chi^2(S;Y) = \tr(\bB \bA)-1,
\end{equation}
where, using \eqref{eq:Qdefn},
\begin{align*}
\bA \defined \bQ_{X,Y}\bQ_{X,Y}^T,\ \bB \defined \bQ_{S,X}^T\bQ_{S,X}.
\end{align*}
\end{lem}
\begin{proof}
See Appendix~\ref{subsec::Lem_trace}.
\end{proof}

The following lemma characterizes some properties of the $\chi^2$-privacy-utility function.
\begin{lem}
\label{lem:concave}
For a given joint distribution $P_{S,X}$, the $\chi^2$-privacy-utility function $F_{\chi^2}(\epsilon;P_{S,X})$ is a concave function in $\epsilon$. Furthermore, $\epsilon \to \frac{1}{\epsilon}F_{\chi^2}(\epsilon;P_{S,X})$ is a non-increasing mapping.
\end{lem}
\begin{proof}
See Appendix~\ref{subsec::Lem_concave}.
\end{proof}

The $\chi^2$-privacy-utility function has a simple upper bound,
\begin{equation}
\label{eq:simplebound}
F_{\chi^2}(\epsilon;P_{S,X}) \leq \epsilon + |\mathcal{X}| -1 -\chi^2(S;X),
\end{equation}
which follows immediately from the data-processing inequality:
\begin{equation}
\chi^2(S;X)+\chi^2(X;Y)\leq \chi^2(S;Y)+\chi^2(X;X).
\end{equation}
We derive an upper bound for the $\chi^2$-privacy-utility function that significantly improves \eqref{eq:simplebound} by using properties of the PICs. The bound is piecewise linear, where each piece has a slope given in terms of a PIC of $P_{S,X}$. Intuitively, this bound corresponds to the privacy-assuring mapping $P_{Y|X}$ that achieves the best PUT if $P_{Y|X}$ was not constrained to be non-negative. We also provide a lower bound that follows directly from the concavity of the $\chi^2$-privacy-utility function. These bounds are illustrated in Fig.~\ref{fig:bound}.

\begin{figure}[!tb]
  \centering
  \includegraphics[width=0.39\textwidth]{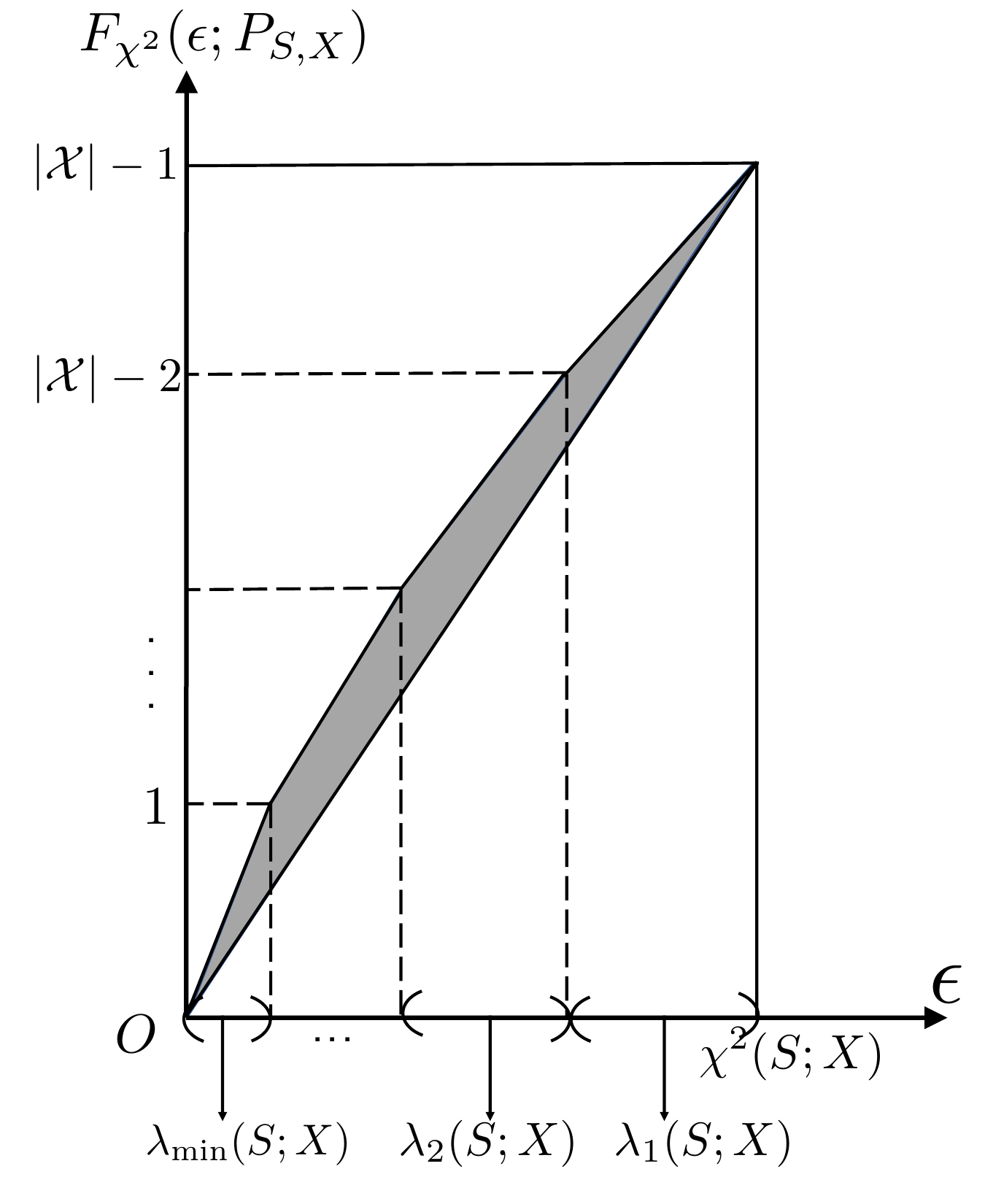}
  \caption{\small{Piecewise linear upper bound and lower bound for the $\chi^2$-privacy-utility function when $\delta(P_{S,X})$, defined in \eqref{equa:delta}, is positive.}}\label{fig:bound}
\end{figure}

\begin{defn}
\label{defn:bound}
For $t_i \in [0,1]\ (i\in [n])$, $  0\leq\epsilon\leq\sum_{i\in [n]} t_i$, and $n\leq m$, $G^m_{\epsilon}(t_1,...,t_n)$ is defined as
\begin{equation*}
\begin{aligned}
G^m_{\epsilon}(t_1,...,t_n) 
\defined \max \left\{ \sum_{i=1}^m x_i\ \Big|\ (x_1,...,x_m) \in \mathcal{D}^m_{\epsilon}(t_1,...,t_n)\right\},
\end{aligned}
\end{equation*}
where 
\begin{equation*}
\begin{aligned}
\mathcal{D}^m_{\epsilon} (t_1,...,t_n)
\defined \left\{(x_1,...,x_m)\ \Big|\ \sum_{i=1}^n t_i x_i  \leq \epsilon, x_i \in [0,1], i\in[m]\right\}.
\end{aligned}
\end{equation*}
\end{defn}

For fixed $m$ and $t_i$ ($i\in[n]$), $G^m_{\epsilon}(t_1,...,t_n)$ is a piecewise linear function with respect to $\epsilon$ and can be expressed in closed-form (cf. Appendix~\ref{subsec::closed_form_G}).

\begin{thm}
\label{thm:bound}
For the $\chi^2$-privacy-utility function $F_{\chi^2}(\epsilon;P_{S,X})$ introduced in Definition~\ref{def:utility-privacy} and $\epsilon \in [0,\chi^2(S;X)]$,
\begin{align*}
\label{inequ:bound}
\frac{|\mathcal{X}|-1}{\chi^2(S;X)}\epsilon \leq F_{\chi^2}(\epsilon;P_{S,X}) \leq G^{|\mathcal{X}|-1}_{\epsilon}(\lambda_1(S;X),...,\lambda_d(S;X)),
\end{align*}
where $d\defined \min\{|\mathcal{S}|,|\mathcal{X}|\}-1$ and $\lambda_1(S;X),...,\lambda_d(S;X)$ are the PICs of $P_{S,X}$.
\end{thm}
\begin{proof}
See Appendix~\ref{subsec::thm_bound_put}.
\end{proof}

\begin{rem}
The upper bound for the $\chi^2$-privacy-utility function given in Theorem~\ref{thm:bound} can also be proved by, for example, combining Theorem~4 in \cite{wang1992some} with properties of the PICs.
\end{rem}

We now illustrate the piecewise linear upper bound. Recall that the PIC decomposition of $P_{S,X}$ results in a set of basis functions $\mathcal{P}\defined \{f_1(S),\cdots,f_d(S)\}$, with corresponding MMSE estimators $\mathcal{U} \defined \{g_1(X),\cdots,g_d(X)\}$.
Consider the following intuition for designing a sequence of privacy-assuring mappings. The first mapping enables the function $g_{d}(X)$ to be reliably estimated from $Y$ while keeping all other functions in $\mathcal{U}$ secret. In this case, the utility is one, since exactly one zero-mean, unit-variance function of $X$ can be recovered from $Y$. The privacy leakage is $\lambda_{d}(S;X)$, since using $g_{d}(X)$ to estimate the private function $f_{d}(S)$ has mean-squared error $1-\lambda_{d}(S;X)$. Following the same procedure, the second privacy-assuring mapping allows only $g_{d}(X)$ and $g_{d-1}(X)$ to be recovered from the disclosed variable and so on. 
This sequence of privacy-assuring mappings corresponds to the breakpoints of the upper bound. Note that such privacy-assuring mappings may not be feasible --- hence the upper bound.

Note that $F_{\chi^2}(0;P_{S,X})$ characterizes the maximal aggregate MMSE of estimating useful functions while guaranteeing perfect privacy. Here perfect privacy means that no zero-mean, unit-variance function of $S$ can be reconstructed from $Y$. If the value of $F_{\chi^2}(0;P_{S,X})$ is known, a better lower bound can be obtained from the concavity of $F_{\chi^2}(\epsilon;P_{S,X})$ as
\begin{equation}
\frac{|\mathcal{X}|-1-F_{\chi^2}(0;P_{S,X})}{\chi^2(S;X)}\epsilon + F_{\chi^2}(0;P_{S,X}) \leq F_{\chi^2}(\epsilon;P_{S,X}).
\end{equation}

When $S=X$, then $\chi^2(S;X) = |\mathcal{X}|-1$ and $F_{\chi^2}(\epsilon;P_{S,X})=\epsilon$. Following from Definition~\ref{defn:bound} and noticing that all PICs of $P_{S,X}$ are 1, the upper bound and the lower bound for the $\chi^2$-privacy-utility function in Theorem~\ref{thm:bound} are both $\epsilon$, which is equal to $F_{\chi^2}(\epsilon;P_{S,X})$. In this sense, the upper bound and lower bound given in Theorem~\ref{thm:bound} are sharp. We investigate the tightness of the upper bound through numerical example in Section~\ref{sec:ParityBits}.

The following corollary of Lemma~\ref{lem:concave} and Theorem~\ref{thm:bound} shows that the $\chi^2$-privacy-utility function is strictly increasing with respect to $\epsilon$.

\begin{cor}
\label{cor:increasing}
For a given joint distribution $P_{S,X}$, the mapping $\epsilon\to F_{\chi^2}(\epsilon;P_{S,X})$ is strictly increasing for $\epsilon \in [0,\chi^2(S;X)]$.
\end{cor}
\begin{proof}
See Appendix~\ref{subsec::cor_increasing_PUT}.
\end{proof}

We denote
\begin{align*}
\partial \mathcal{D}(\epsilon;P_{S,X}) \defined \{P_{Y|X} \mid S\rightarrow X\rightarrow Y, \chi^2(S;Y) = \epsilon\}.
\end{align*}
By Corollary~\ref{cor:increasing}, $F_{\chi^2}(\epsilon;P_{S,X})$ is strictly increasing. Therefore,
\begin{equation}
F_{\chi^2}(\epsilon;P_{S,X})
= \max_{P_{Y|X}\in \partial \mathcal{D}(\epsilon;P_{S,X})}\chi^2(X;Y).
\end{equation}

By Corollary 7 in \cite{calmon2017principal}, when $\delta(P_{S,X})=0$, defined in \eqref{equa:delta}, then $F_{\chi^2}(0;P_{S,X})>0$ (i.e., there exists a privacy-assuring mapping that allows the disclosure of a non-trivial amount of useful functions while guaranteeing perfect privacy). On the other hand, when $\delta(P_{S,X})>0$, then $F_{\chi^2}(0;P_{S,X})=0$. The following theorem shows that when $\delta(P_{S,X})>0$, the upper bound of $F_{\chi^2}(\epsilon;P_{S,X})$ in Theorem~\ref{thm:bound} is achievable around zero, implying that the upper bound is tight around zero. The proof of this theorem also provides a specific way to construct an optimal privacy-assuring mapping (i.e., achieves the upper bound in Theorem~\ref{thm:bound}).

\begin{thm}
\label{thm:highprivacyregion}
Suppose $\delta(P_{S,X})>0$ and $P_{X\min}>0$. Then there exists $Y$ such that $S\rightarrow X\rightarrow Y$, $\chi^2(X;Y) = P_{X\min}$ and $\chi^2(S;Y) = P_{X\min}\lambda_{\min}(S;X)$.
\end{thm}
\begin{proof}
See Appendix~\ref{subsec::thm_high_pri_region}.
\end{proof}

When $\delta(P_{S,X})>0$ and $P_{X\min}>0$, then $F_{\chi^2}(\hat{\epsilon};P_{S,X}) = P_{X\min}$ where $\hat{\epsilon}=P_{X\min}\lambda_{\min}(S;X)$. Since $(\hat{\epsilon},P_{X\min})$ is a point on the upper bound of the $\chi^2$-privacy-utility function given in Theorem~\ref{thm:bound}, Theorem~\ref{thm:highprivacyregion} shows that, in this case, the upper bound is achievable in the high-privacy region. We remark that the local behavior of privacy-utility functions in high-privacy region and high-utility region has been studied in the context of strong data processing inequalities (e.g., \cite{du2018strong,makur2017linear} and the references therein).

\subsection*{Connections with Maximal Correlation}

Maximal correlation has previously been considered as a privacy measure in \cite{asoodeh2016privacy,asoodeh2016information,li2018maximal,calmon2013bounds}. In particular, it has been proved \cite{calmon2013bounds} that when $\rho_m(S;Y)$ is small, then $\Pr(S\neq \hat{S})$ can be lower bounded for any $\hat{S}=h(Y)$. 
We show in Corollary~\ref{cor::rhom_max_leak_eq} that, in the high privacy regime, the privacy-utility function under maximal correlation possesses similar properties to $F_{\chi^2}(\epsilon;P_{S,X})$. 
However, when $\rho_m(S;Y)$ is large, say $\rho_m(S;Y)=1$, it is unclear whether one private function or several private functions can be recovered from the disclosed variable. In contrast, $\chi^2$-information can distinguish between these two cases and quantifies how many private functions, on average, can be reconstructed from the disclosed variable. For example, a user might be comfortable revealing that his/her age is above a certain threshold, but not the age itself. In this case, the privacy leakage measured by maximal correlation is one since there is a function of age which can be recovered from the disclosed variable. Thus, maximal correlation cannot distinguish between the cases where only one function of $S$ and $S$ itself can be estimated from the disclosed data.
We will revisit this example in the next section and show how to design privacy-assuring mappings using PICs which target specific private functions and useful functions. 
Finally, we provide an example showing the limitation of maximal correlation as a utility measure. 
\begin{defn}
\label{def:rhom_utility-privacy}
For a given joint distribution $P_{S,X}$ and $0\leq \epsilon \leq \rho_m(S;X)$, we define the maximal-correlation-privacy-utility (trade-off) function as
\begin{align*}
    F_{\rho_m}(\epsilon;P_{S,X}) \defined \sup_{P_{Y|X}\in \mathcal{D}_{\rho_m}(\epsilon;P_{S,X})}\rho_m(X;Y),
\end{align*}
where $\mathcal{D}_{\rho_m}(\epsilon;P_{S,X}) \defined \{ P_{Y|X} \mid S\rightarrow X\rightarrow Y, \rho_m(S;Y)\leq \epsilon\}$.
\end{defn}
The next corollary follows from the same proof techniques used in Theorem~\ref{thm:bound} and Theorem~\ref{thm:highprivacyregion}.
\begin{cor}
\label{cor::rhom_max_leak_eq}
For a given joint distribution $P_{S,X}$ and $\epsilon \in \left[0, \rho_m(S;X)\right]$, if $\delta(P_{S,X})>0$, then $F_{\rho_m}(\epsilon;P_{S,X}) \leq \epsilon/\sqrt{\lambda_{\min}(S;X)}$. 
Furthermore, if $P_{X\min}>0$, then there exists $Y$ such that $S\to X\to Y$, $\rho_m(X;Y)=\sqrt{P_{X\min}}$ and $\rho_m(S;Y)=\sqrt{P_{X\min}\lambda_{\min}(S;X)}$.
\end{cor}
We illustrate the limitation of the maximal correlation as a utility measure through the following example.
\begin{example}
\label{ex::max_cor_pic_compare}
Let $\mathcal{S}=\{-1,1\}^n$ and $\mathcal{X}=\{-1,1\}^n$, and $X^n$ be the result of passing $S^n$ through a memoryless binary symmetric channel with crossover probability $\epsilon<1/2$. We assume that $S^n$ is composed of $n$ uniform and i.i.d. bits. For $\mathcal{A}\subseteq [n]$, let $Y=\prod_{i\in \mathcal{A}}X_i$. In this case, one can show that $\rho_m(S^n;Y) = (1-2\epsilon)^{|\mathcal{A}|}$ and $\rho_m(X^n;Y) = 1$. If $|\mathcal{A}|$ is an increasing function of $n$, then $\rho_m(S^n;Y)\to 0$ as $n\to \infty$. In other words, we can disclose a function of $X^n$ achieving nearly perfect privacy and utility as measured by $\rho_m(S^n;Y)$ and $\rho_m(X^n;Y)$, respectively, with large $|\mathcal{A}|$ and $n$. 
However, as $n$ increases, the basis of functions in $\mathcal{L}_2(P_{X^n})$ will increase exponentially, and revealing only one function may not be enough for achieving utility. 
The crux of the limitation is that maximal correlation only takes into account the most reliably estimated function.
The $\chi^2$-information overcomes this limitation by capturing all possible real-valued functions of $X^n$ that can be recovered from $Y$.
In particular, if $\chi^2(X^n;Y)=|\mathcal{X}|-1$, then all zero-mean finite-variance functions of $X^n$ can be reconstructed from $Y$.
We will revisit this example again in Section~\ref{sec:MMSE} and Section~\ref{sec:numericalresults}.
\end{example}
\section{Composite PUTs:\\
A Convex Program for Computing Privacy-Assuring Mappings}
\label{sec:add_opt}
In the previous section, we studied $\chi^2$-based metrics for both privacy and utility. The optimization problem in the definition of $\chi^2$-privacy-utility function (Definition~\ref{def:utility-privacy}) is non-convex.
Next, we provide a convex program for designing privacy-assuring mappings by adding more stringent constraints on privacy and utility.

More specifically, we explore an alternative, finer-grained approach for measuring both privacy and utility based on PICs (recall that $\chi^2$-information is the sum of all PICs). This approach has a practical motivation, since oftentimes there are specific well-defined features (functions) of the data (realizations of a random variable) that should be hidden or disclosed. For example, a user may be willing to disclose that they prefer documentaries over action movies, but not exactly which documentary they like. More abstractly, we consider the case where certain known functions should be disclosed (utility), whereas others should be hidden (privacy). This is a finer-grained setting than the one used in the last section, since $\chi^2$-information captures the aggregate reconstruction error across all zero-mean, unit-variance functions.

We denote the set of functions to be disclosed as 
\begin{align*}
\mathcal{U}(X)\defined \{u_i:\mathcal{X}\to\Reals \mid \EE{u_i(X)}=0,
||u_i(X)||_2=1, i\in[n]\},
\end{align*}
and the set of functions to be hidden as 
\begin{align*}
\mathcal{P}(S)\defined \{s_i:\mathcal{S}\to\Reals \mid \EE{s_i(S)}=0, 
||s_i(S)||_2=1, i \in[m]\}.
\end{align*}
Our goal is to find the privacy-assuring mapping $P_{Y|X}$ such that $S\rightarrow X \rightarrow Y$ and $Y$ satisfies the following privacy-utility constraints:
\begin{enumerate}
\item \textbf{Utility constraints:} $\max \{\mmse(u_i(X)|Y)\}_{i\in [n]} \leq \Delta$ and $X\sim Y$.
\item \textbf{Privacy constraints:} $\mmse(s_i(S)|Y) \geq \theta_i, i\in[m]$.\label{enu_pri_cons}
\end{enumerate}
Note that the utility constraint $X\sim Y$ implies that the disclosed variable follows the same distribution as the useful variable. The practical motivation for adding this constraint is to enable $Y$ to preserve overall population statistics about $X$, while hiding information about individual samples. This assumption also enables the problem of finding the optimal privacy-assuring mapping to be formulated as a convex program, described next.

We follow two steps -- projection\footnote{We call this step as projection because of the geometric interpretation of conditional expectation (see, e.g., \cite{durrett2010probability}).} and optimization -- to find the privacy-assuring mapping. Private functions are projected to a new set of functions based on the useful variable in the first step. Then a PIC-based convex program is proposed in order to find the privacy-assuring mapping.

\subsection{Projection}
As a first step, we project (i.e., compute the conditional expectation) all private functions to the useful variable and obtain a new set of functions:
\begin{align*}
    \mathcal{P}(X)\defined \left\{ \hat{s}_i(x) \defined \frac{\EE{s_i(S)|X=x}}{||\EE{s_i(S)|X}||_2}\ \Big|\ i\in[m] \right\}.
\end{align*}
It is worth noting that, after the projection, the obtained privacy-assuring mapping may not be an optimal solution to the original problem since the privacy constraints become stricter (see Lemma~\ref{lem:project}). Nonetheless, the advantage of this projection is twofold. First, it can significantly simplify the optimization program, since after the projection all functions are cast in terms of the useful variable alone. Second, the private variable is not needed as an input to the optimization after the projection. Therefore, the party that solves the optimization does not need access to the private data directly, further guaranteeing the safety of the sensitive information. The following lemma proves that privacy guarantees cast in terms of the projected functions still hold for the original functions.
\begin{lem}
\label{lem:project}
Assume $S\rightarrow X \rightarrow Y$. For any function $f:\mathcal{S} \to \Reals$, if $\EE{f(S)} = 0$ and $||\EE{f(S)|X}||_2\neq 0$, we have $\EE{\EE{f(S)|X}}=0$ and
\begin{align*}
\mmse\left(\frac{f(S)}{||f(S)||_2}\Bigg|Y\right)\geq \mmse\left(\frac{\EE{f(S)|X}}{||\EE{f(S)|X}||_2}\Bigg|Y\right).
\end{align*}
\end{lem}
\begin{proof}
See Appendix~\ref{subsec::thm_projection}.
\end{proof}

By Lemma~\ref{lem:project}, $\mmse(s_i(S)|Y)\geq \mmse(\hat{s}_i(X)|Y)$. Therefore, if the new set of functions satisfies the privacy constraints (i.e., $\mmse(\hat{s}_i(X)|Y)\geq \theta_i$), the original set of functions also satisfies the privacy constraints (i.e., $\mmse(s_i(S)|Y)\geq \theta_i$).

\subsection{Optimization}
We introduce next a PIC-based convex program to find the privacy-assuring mapping $P_{Y|X}$. First, we construct a matrix $\mathbf{F}$ given by $(\mathbf{f}_0,\mathbf{f}_1,...,\mathbf{f}_{|\mathcal{X}|-1})$ such that
\begin{equation}
\label{equ:F_orth}
\mathbf{F}^T\mathbf{D}_X\mathbf{F}=\mathbf{I},
\end{equation}
\begin{equation}
\label{equ:span_fandu}
\mathsf{span}(\{\mathbf{f}_0,...,\mathbf{f}_{n'}\})=\mathsf{span}(\{\mathbf{f}_0,\mathbf{u}_1,...,\mathbf{u}_n\}),
\end{equation}
where $\mathbf{f}_0 \defined (1,...,1)^T$, $\mathbf{f}_i \defined (f_i(1),...,f_i(|\mathcal{X}|))^T$, $\mathbf{u}_i \defined (u_i(1),...,u_i(|\mathcal{X}|))^T$, and
\begin{align*}
    n'\defined \mathsf{dim}(\mathsf{span}(\{\mathbf{f}_0,\mathbf{u}_1,...,\mathbf{u}_n\}))-1.
\end{align*}

Following from \eqref{equ:F_orth},
$\{f_k(x)\mid k=0,...,|\mathcal{X}|-1\}$ is a basis of $\mathcal{L}_2(P_X)$ and, consequently, the functions $\hat{s}_i(x)$ can be decomposed as
\begin{equation}
\hat{s}_i(x) = \sum_{k=0}^{|\mathcal{X}|-1} \alpha_{i,k} f_k(x).
\end{equation}
Since $\EE{\hat{s}_i(X)}=0$, then $\alpha_{i,0}=0$. Similarly, since $\mathbf{u}_i\in \mathsf{span}(\{\mathbf{f}_0,...,\mathbf{f}_{n'}\})$ and $\EE{u_i(X)}=0$, we have
\begin{equation}
u_i(x) = \sum_{k=1}^{n'} \beta_{i,k} f_k(x).
\end{equation}
If $\mathbf{P}_{X,Y}=\mathbf{D}_X \mathbf{F} \mathbf{\Sigma} \mathbf{F}^T \mathbf{D}_X$ with $\mathbf{\Sigma}=\diag(1,\sigma_1,...,\sigma_{|\mathcal{X}|-1})$ is a feasible joint distribution matrix (i.e., non-negative entries and all entries add to $1$), then, following from Theorem~\ref{thm:PIC_Charac},
\begin{align*}
\mmse(\hat{s}_i(X)|Y) 
&=1-\sum_{k=1}^{|\mathcal{X}|-1} \alpha^2_{i,k} \sigma_k^2,\\
\mmse(u_i(X)|Y) 
&= \sum_{k=1}^{n'} \beta^2_{i,k} (1-\lambda_k(X;Y))
\leq 1-\min_{k\in[n']}\lambda_k(X;Y)
=1 - \left(\min_{k\in[n']} \sigma_k\right)^2.
\end{align*}
Therefore, the design of the privacy-assuring mapping $P_{Y|X}$ with privacy-utility constraints is equivalent to solving the PIC-based convex program in Formulation~\ref{alg:PIC-based}. In this case, the objective function is chosen as $\mathsf{obj}(\sigma_1,...,\sigma_{n'})=\min\{\sigma_1,...,\sigma_{n'}\}$\footnote{This is a convex program since one can add a constraint $\sigma_i\geq \sigma$ ($i\in [n']$) and maximize $\sigma$.}.

\begin{algorithm}[t]
\label{alg:PIC-based}
\begin{align}
    \max~&\mathsf{obj}(\sigma_1,...,\sigma_{n'})\\
    ~\sto~&\sum_{k=1}^{|\mathcal{X}|-1} \alpha_{i,k}^2 \sigma_k^2 \leq 1-\theta_i\ (i=1,...,m),\\
    &0\leq \sigma_i \leq 1\ (i=1,...,|\mathcal{X}|-1),\\
    &\mathbf{\Sigma} = \diag(1,\sigma_1,...,\sigma_{|\mathcal{X}|-1}),\\
    &\mathbf{P}_{X,Y}=\mathbf{D}_X \mathbf{F} \mathbf{\Sigma} \mathbf{F}^T \mathbf{D}_X,\\
    &\mathbf{P}_{X,Y}\ \text{has non-negative entries.}
\end{align}
\caption{\small{PIC-based convex program. Here $\sigma_i$ ($i=1,...,|\mathcal{X}|-1$) and $\theta_i$ ($i=1,...,m$) are variables and privacy parameters, respectively. The objective function $\mathsf{obj}(\sigma_1,...,\sigma_{n'})$ is chosen as a concave function and measures utility.}}
\end{algorithm}

The objective function $\min\{\sigma_1,...,\sigma_{n'}\}$ maximizes the worst-case utility over all useful functions. On the other hand, we can choose the objective function to be a weighted sum  $\sum_{i=1}^{n'} a_i\sigma_i$. Although maximizing the weighted sum is not equivalent to the desired utility constraints, this new formulation allows more flexibility in the optimization. In particular, this enables useful functions which do not highly correlate with private functions to achieve better utility, in terms of mean-squared error, under the same privacy constraints. Furthermore, the weights can be used to prioritize the reconstruction of certain useful functions.

The previous convex programs can be numerically solved by standard methods (e.g., CVXPY \cite{diamond2016cvxpy}). Note that when all useful functions and private functions are based on the same random variable, we can use optimization without projection. We defer the numerical results to Section~\ref{sec:numericalresults}, where we derive privacy-assuring mappings for a synthetic dataset and a real-world dataset using tools introduced in this section.
\section{Lower Bounds for MMSE with Restricted Knowledge of the Data Distribution}
\label{sec:MMSE}
So far we have assumed the information-theoretic setting where the probability distribution $P_{S,X}$ is known to the privacy mechanism designer beforehand. In this section, we forgo this assumption and consider a setting where $S=\phi(X)$ and the correlation between $\phi(X)$ and a set of functions (composed with the data) $\{\phi_j(X)\}_{j=1}^m$ is given.  We derive lower bounds for the MMSE of estimating $\phi(X)$ given $Y$ in terms of the MMSE of estimating $\phi_j(X)$ given $Y$. In privacy systems, $X$ may be a user's data and $Y$ a distorted version of $X$ generated by a privacy-assuring mapping $P_{Y|X}$. The set $\{\phi_j(X)\}_{j=1}^m$ could then represent a set of functions that are known to be hard to infer from $Y$ due to inherent privacy constraints of the setup. For example, when the mapping $P_{Y|X}$ is designed by the PIC-based convex programs in Formulations~\ref{alg:PIC-based} and $\{\phi_j(X)\}_{j=1}^m$ is the set of private functions, $\mmse\left(\phi_j(X)|Y\right)$ is lower bounded due to the privacy constraints.

The following lemma will be used to derive the lower bounds for the MMSE of $\phi(X)$ given $Y$.
\begin{lem}
\label{lem:quadBound}
Let $L_n:(0,\infty)^n\times[0,1]^n\rightarrow \Reals$ be given by
\begin{equation}
L_n(\ba,\bb) \defined \max\left\{ \ba^T\by \mid \by\in\Reals^n,\|\by\|_2\leq 1, \by\leq \bb \right\}. 
\label{eq:defn_Ln}
\end{equation}
Let $\pi$ be a permutation of $[n]$ such that $b_{\pi(1)}/a_{\pi(1)}\leq \dots\leq b_{\pi(n)}/a_{\pi(n)}$. If $b_{\pi(1)}/a_{\pi(1)}\geq 1$, $L_n(\ba,\bb)=\|\ba\|_2$. Otherwise,
\begin{align*}
L_n(\ba,\bb)
=\sum_{i=1}^{k^*} a_{\pi(i)}b_{\pi(i)} + \sqrt{\left(\|\ba\|_2^2-\sum_{i=1}^{k^*} a_{\pi(i)}^2\right)\left(1-\sum_{i=1}^{k^*}b_{\pi(i)}^2\right)}
\end{align*}
where
\begin{equation}
\label{eq:kstar}
k^*\defined \max\left\{k\in [n]\ \Big|\  
\frac{b_{\pi(k)}}{a_{\pi(k)}}\leq \sqrt{\frac{\left(1-\sum_{i=1}^{k-1}b_{\pi(i)}^2\right)^+}{\|\ba\|_2^2-\sum_{i=1}^{k-1} a_{\pi(i)}^2}} \right\}.
\end{equation}
\end{lem}
\begin{proof}
See Appendix~\ref{subsec::Lem_quadBound}.
\end{proof}

Throughout this section we assume $\|\phi_i(X)\|_2=1$ ($i\in [m]$) and $\EE{\phi_i(X)\phi_j(X)}=0$ ($i\neq j$).
For a given $\phi_i$, the inequality 
\begin{equation}
\label{equa:condE_2leqnu}
  \max_{\psi\in\calL_2(P_Y)} \EE{\phi_i(X)\psi(Y)}=\|\EE{\phi_i(X)|Y} \|_2\leq \nu_i
\end{equation}
is satisfied, where $0\leq \nu_i\leq 1$. This is equivalent to $\mmse(\phi_i(X)|Y)\geq 1-\nu_i^2$.

\begin{thm}
\label{thm:loose}
Let $\|\phi(X)\|_2=1$ and  $\EE{\phi(X)\phi_i(X)}=\rho_i>0$. Denoting $\brho\defined (|\rho_1|,\dots,|\rho_m|)$, $\bnu\defined (\nu_1,\dots,\nu_m)$, $\rho_0\defined\sqrt{1-\sum_{i=1}^m\rho_i^2} $, $\brho_0\defined (\rho_0,\brho)$ and $\bnu_0\defined(1,\bnu)$, then
\begin{equation}
\label{eq:firstBound}
\|\EE{\phi(X)|Y} \|_2\leq B_m(\brho_0,\bnu_0),
\end{equation}
where
\begin{equation}
\label{eq:Bmdef}
B_m(\brho_0,\bnu_0)\defined
\begin{cases}
L_{m+1}\left( \brho_0,\bnu_0 \right)  & \mbox{if } \rho_0>0,\\
L_m(\brho,\bnu) & \mbox{otherwise,}  
\end{cases}
\end{equation}
and $L_{n}$ is given in \eqref{eq:defn_Ln}. Consequently,
\begin{equation}
\label{eq:consfirstBound_mmse}
\mmse(\phi(X)|Y)\geq 1-B_m(\brho_0,\bnu_0)^2.
\end{equation}
\end{thm}
\begin{proof}
See Appendix~\ref{subsec::thm_loose}.
\end{proof}

Denote $\psi_i(Y)\defined (T_{X|Y} \phi_i)(Y)/\|(T_{X|Y} \phi_i)(Y)\|_2$ ($i\in [m]$) and $\phi_0(X) \defined \rho_0^{-1}(\phi(X)-\sum_{i=1}^m\rho_i\phi_i(X))$ if $\rho_0>0$, otherwise $\phi_0(X) \defined 0$. The previous bounds, \eqref{eq:firstBound} and \eqref{eq:consfirstBound_mmse}, can be further improved when $\EE{\psi_i(Y)\phi_j(X)}=0$ for $i\neq j,\, j\in \{0,\dots,m\}$.

\begin{thm}
\label{thm:tighter}
Let $\|\phi(X)\|_2=1$ and $|\EE{\phi(X)\phi_i(X)}|=\rho_i>0$ for $i\in[m]$. In addition, assume $\EE{\psi_i(Y)\phi_j(X)}=0$ for $i\neq j$, $i\in [t]$ and $j\in\{0,\dots,m\}$, where $0\leq t\leq m$. Then
\begin{equation}
\|\EE{\phi(X)|Y} \|_2\leq \sqrt{\sum_{k=1}^t \nu_i^2\rho_i^2 + B_{m-t}\left(\tilde{\brho},\tilde{\bnu}\right)^2},
\end{equation}
where $\tilde{\brho}=(\rho_0,\rho_{t+1},\dots,\rho_m)$, $\tilde{\bnu}=(1,\nu_{t+1},\dots,\nu_m)$, and $B_m$ is defined in \eqref{eq:Bmdef} (considering $B_0=0$). In particular, if $t=m$,
\begin{equation}
\label{eq:sharp}
\|\EE{\phi(X)|Y} \|_2\leq  \sqrt{\rho_0^2+\sum_{k=1}^m \nu_i^2\rho_i^2 },
\end{equation}
and \eqref{eq:sharp} is an equality when $\rho_0=0$. Furthermore,
\begin{equation}
\label{eq:mmse_sharp}
\mmse(\phi(X)|Y)\geq 1-\sum_{k=1}^t \nu_i^2\rho_i^2- B_{m-t}\left(\tilde{\brho},\tilde{\bnu}\right)^2.
\end{equation}
\end{thm}
\begin{proof}
See Appendix~\ref{subsec::thm_tighter}.
\end{proof}

In what follows, we use three examples to illustrate different use cases of Theorem \ref{thm:loose} and \ref{thm:tighter}. Example~\ref{ex:sym_channel_bound_sharp} illustrates how Theorem~\ref{thm:tighter} can be applied to the $q$-ary symmetric channel which could be perceived as a model of randomized response \cite{warner1965randomized,kairouz2014extremal}, and demonstrates that bound~\eqref{eq:sharp} is sharp. Example~\ref{ex:bin_chan_parity_bits_mmse} illustrates Theorem~\ref{thm:tighter} for the binary symmetric channel. Here the useful variable is composed by $n$ uniform and independent bits. In this case, the basis can be expressed as the parity bits of the input to the channel. Finally, Example~\ref{ex:one_bit_func_mmse} illustrates Theorem~\ref{thm:loose} for one-bit functions. The same method used in the proof of Theorem~\ref{thm:loose} is applied to bound the probability of correctly guessing a one-bit function from an observation of the disclosed data.

\begin{example}[$q$-ary symmetric channel]
\label{ex:sym_channel_bound_sharp}
Let $\calX=\calY=[q]$, and $Y$ be the result of passing $X$ through an $(\epsilon,q)$-ary symmetric channel, which is defined by the transition probability
\begin{equation}
P_{Y|X}(y|x)= (1-\epsilon)\indicator_{y=x} +\epsilon/q \text{ for all } x\in\mathcal{X},y\in\mathcal{Y}.
\end{equation}
We assume that $X$ has a uniform distribution, which implies $Y$ also has a uniform distribution. Any function $\phi\in\calL_2(P_X)$ such that $\EE{\phi(X)}=0$ and $\|\phi(X)\|_2=1$ satisfies
\begin{align*}
\psi(Y)=(T_{X|Y}\phi)(Y)=(1-\epsilon)\phi(Y),
\end{align*}
and, consequently, $\|(T_{X|Y}\phi)(Y)\|_2 = (1-\epsilon)$. We will use this fact to show that the bound \eqref{eq:sharp} is sharp in this case.
    
Observe that for  $\phi_i,\phi_j\in \calL_2(P_X)$, if $\EE{\phi_i(X)\phi_j(X)}=0$ then $\EE{\psi_i(Y)\psi_j(Y)}=0$. Now let $\phi\in \calL_2(P_X)$ satisfy $\EE{\phi(X)}=0$ and $\|\phi(X)\|_2=1$, and let $\EE{\phi(X)\phi_i(X)}=\rho_i$ for $i\in [m]$, where $\{\phi_i\}$ satisfies the conditions in Theorem~\ref{thm:tighter} and $\sum_{i=1}^m \rho_i^2=1$. In addition, $\|\psi_i(Y)\|_2=(1-\epsilon)=\nu_i$. Then, from \eqref{eq:sharp} and noting that $\rho_0=0$, $t=m$, we have
\begin{align*}
\|(T_{X|Y}\phi)(Y)\|_2 \leq \sqrt{\sum_{i=1}^m \nu_i^2\rho_i^2}
                = (1-\epsilon)\sqrt{\sum_{i=1}^m \rho_i^2}
                = 1-\epsilon,
\end{align*}
which matches $\|(T_{X|Y}\phi)(Y)\|_2$, and the bound is tight in this case.
\end{example}

\begin{example}[Binary channels with additive noise]
\label{ex:bin_chan_parity_bits_mmse}
Let $\calX = \{-1,1\}^n$ and $\calY=\{-1,1\}^n$, and $Y^n$ be the result of passing $X^n$ through a memoryless binary symmetric channel with crossover probability $\epsilon<1/2$. We assume that $X^n$ is composed by $n$ uniform and i.i.d. bits. For $\calS\subseteq [n]$, let
\begin{align*}
    \chi_\calS(X^n)\defined\prod_{i\in \calS}X_i.
\end{align*}
Any function $\phi:\calX \to \Reals$ can then be decomposed in terms of the basis $\chi_\calS(X^n)$ as \cite{o2008some}
\begin{equation*}
\phi(X^n)=\sum_{\calS\subseteq[n]}c_\calS \chi_{\calS}(X^n),
\end{equation*}
where $c_\calS=\EE{\phi(X^n)\chi_\calS(X^n)}$. Furthermore, since $\EE{\chi_\calS(X^n)|Y^n}=(1-2\epsilon)^{|\calS|}\chi_\calS(Y^n)$, it follows from Theorem~\ref{thm:tighter} that
\begin{equation}
\mmse(\phi(X^n)|Y^n) =1-\sum_{\calS\subseteq[n]} c_\calS^2(1-2\epsilon)^{2|\calS|}.
\end{equation} 
This result can be generalized for the case $X^n=Y^n\otimes Z^n$, where the operation $\otimes$ denotes bit-wise multiplication, $Z^n$ is drawn from $\{-1,1\}^n$ and $X^n$ is uniformly distributed. In this case
\begin{equation}
\mmse(\phi(X^n)|Y^n) =1-\sum_{\calS\subseteq[n]}
c_\calS^2\EE{\chi_\calS(Z^n)}^2.
\end{equation}
\end{example}

\begin{example}[One-Bit Functions]
\label{ex:one_bit_func_mmse}
Let $X$ be a hidden random variable with support $\calX$, and let $Y$ be a noisy observation of $X$. We denote by $B_1,\dots,B_m$ a collection of $m$ predicates of $X$, where $B_i = \phi_i(X)$, $\phi_i:\calX\rightarrow \{-1,1\}$ for $i\in[m]$ and, without loss of generality, $\EE{B_i}=b_i\geq 0$.

We denote by $\hat{B}_i$  an estimate of $B_i$ given an observation of $Y$, where $B_i\rightarrow X \rightarrow Y \rightarrow \hat{B}_i$. We assume that for any $\hat{B}_i$ $$\left|\mathbb{E}[B_i\hat{B}_i]\right|\leq 1-2\alpha_i$$
for some $0\leq \alpha_i\leq (1-b_i)/2 \leq 1/2$. This condition is equivalent to imposing that $\Pr(B_i\neq \hat{B}_i)\geq \alpha_i$, since
\begin{align*}
\EE{B_i\hat{B}_i}
&=\Pr(B_i=\hat{B}_i)-\Pr(B_i\neq \hat{B}_i)\\
&=1-2\Pr(B_i\neq \hat{B}_i).
\end{align*}
In particular, this captures the ``hardness'' of guessing $B_i$ based solely on an observation of $Y$.

Now assume there is a bit $B$ such that $\EE{B B_i}=\rho_i$ for $i\in [m]$ and $\EE{B_iB_j}=0$ for $i\neq j$. We can apply the same method used in the proof of Theorem~\ref{thm:loose} to bound the probability of
$B$ being guessed correctly from an observation of $Y$:
\begin{equation}
\Pr(B\neq \hat{B}) \geq \frac{1}{2}\left( 1-B_m(\brho,\bnu)\right),
\end{equation}
where $\nu_i = 1-2\alpha_i$.
\end{example}
\section{Robustness of the PUTs}
\label{sec:robust}

In this section we investigate the pipeline in Fig.~\ref{fig:robust}  for designing  privacy-assuring mappings in practice. During the training time, a reference dataset with $n$ samples is drawn from $P_{S,X}$. The distribution of the source is estimated by computing the empirical distribution (type) $P_{\hat{S},\hat{X}}$ of the reference dataset. $P_{\hat{S},\hat{X}}$ and the privacy-utility constraints are then used as inputs to a convex program solver that returns the corresponding privacy-assuring mapping $W_{Y|\hat{X}}$ (if feasible). We denote by $\hat{Y}$ the random variable produced by randomizing $\hat{X}$ according to $W_{Y|\hat{X}}$, i.e., by applying the privacy-assuring mapping to a source with distribution $P_{\hat{S},\hat{X}}$. During the testing time, new i.i.d. samples from the source $P_{S,X}$ are randomized using the privacy-assuring mapping $W_{Y|\hat{X}}$ computed during the training time, resulting in the disclosed variable $Y$.

The privacy and utility constraints used for computing the privacy-assuring mapping hold for a data source with distribution $P_{\hat{S},\hat{X}}$, since this is the distribution used as an input to the optimization program. However, during the testing time,  $W_{Y|\hat{X}}$ is applied to new samples from the source $P_{S,X}$. {\it Do the privacy and utility guarantees still hold during the testing time?} Since as $n$ increases $P_{\hat{S},\hat{X}}$ converges to $P_{S,X}$, it is natural to expect that the privacy and utility guarantees during the testing time will not be far from the ones selected during the training time.

In what follows, we analyze the robustness of the PUT optimization using $\chi^2$-information, and characterize the gap between privacy and utility guarantees of the training and testing time in terms of the number of samples in the reference dataset and the probability of less likely symbols. The following lemma will be used to prove the main result in this section.

\begin{figure}[!tb]
  \centering
  \includegraphics[width=0.52\textwidth]{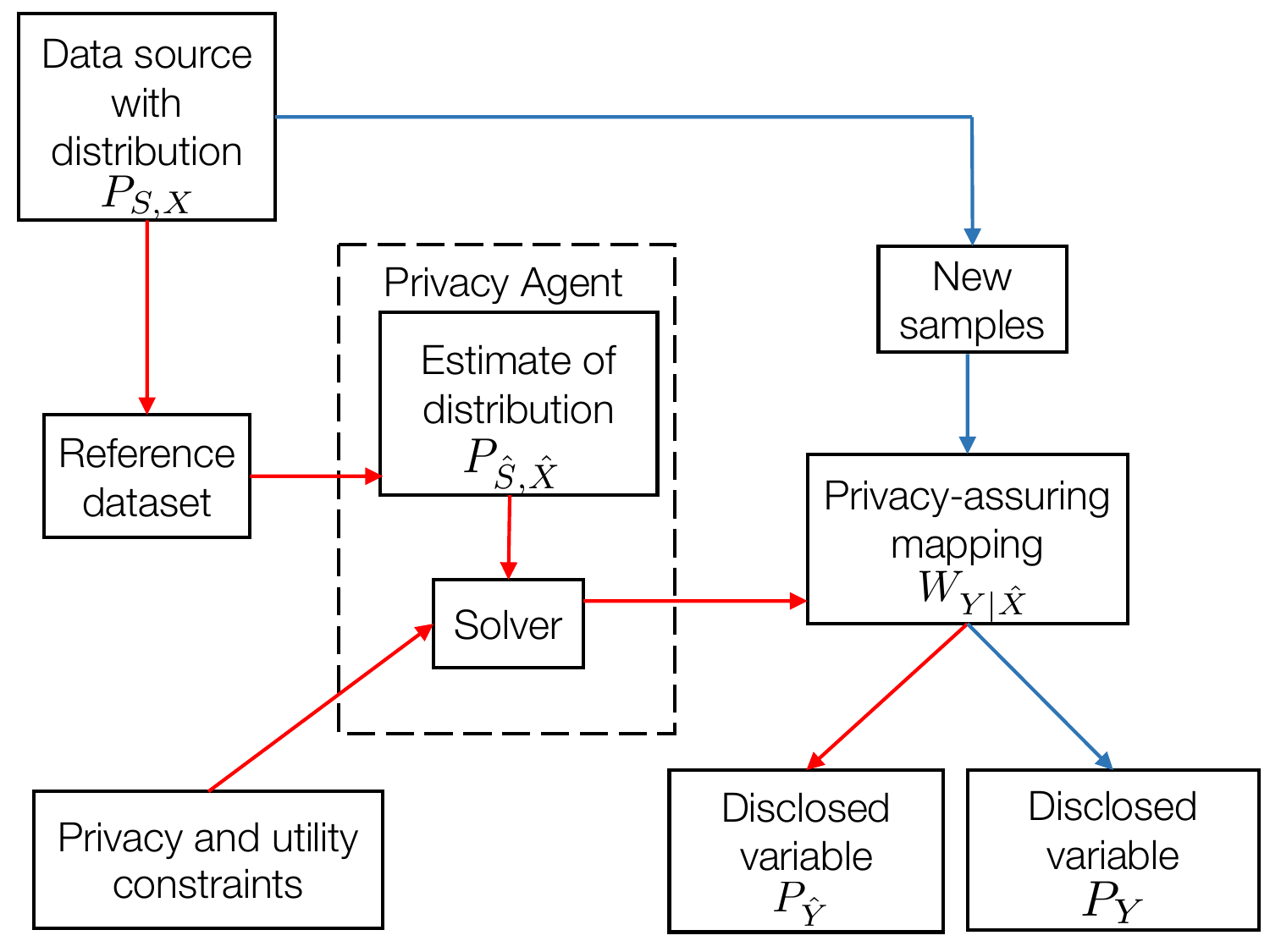}
  \caption{\small{Flowchart for training time (red) and testing time (blue).}}\label{fig:robust}
\end{figure}

\begin{lem}
\label{lem:rob}
Suppose that $S_i\rightarrow X_i \rightarrow Y_i$ for $i=1,2$ and $P_{Y_1|X_1}=P_{Y_2|X_2}$. Let $m_S\defined \min\{P_{S_i}(s)\mid s\in \mathcal{S},i=1,2\}$ and $m_X\defined \min\{P_{X_i}(x)\mid x\in \mathcal{X},i=1,2\}$. Then 
\begin{align*}
|\chi^2(S_1;Y_1) - \chi^2(S_2;Y_2)|
&\leq \frac{4}{m_S}||P_{S_1,X_1}-P_{S_2,X_2}||_1,\\
|\chi^2(X_1;Y_1) - \chi^2(X_2;Y_2)|
&\leq \frac{4}{m_X}||P_{S_1,X_1}-P_{S_2,X_2}||_1.
\end{align*}
\end{lem}
\begin{proof}
See Appendix~\ref{subsec::Lem_rob}.
\end{proof}

Next, we illustrate the sharpness of the upper bounds in Lemma~\ref{lem:rob} through the following example.
\begin{example}
Let $\mathcal{X}=\mathcal{Y}=\{0,1\}$, $S_i=X_i$ for $i\in\{1,2\}$. Assume that $P_{X_1}(1)=p$ and $P_{X_2}(1)=p+\epsilon/2$ are such that $p,\epsilon\in(0,1/2)$. Let $P_{Y|X}$ denote the Z-channel determined by
\begin{equation}
    P_{Y|X}(y|x)=
    \begin{cases}
    a &\text{if } y=1,x=1,\\
    1-a &\text{if } y=0,x=1,\\
    1 &\text{if } y=0,x=0,\\
    \end{cases}
\end{equation}
where $a\in(0,1)$. In this case, we have $\epsilon = \|P_{S_1,X_1}-P_{S_2,X_2}\|_1$ which does not depend on $a$ and $p$, and
\begin{equation*}
\chi^2(S_i;Y_i) 
= \chi^2(X_i;Y_i)
=1-\frac{1-a}{1-aP_{X_i}(1)}, \text{ for }i\in\{1,2\}.
\end{equation*}
We denote $\Delta \defined |\chi^2(S_1;Y_1)-\chi^2(S_2;Y_2)| 
=|\chi^2(X_1;Y_1)-\chi^2(X_2;Y_2)|$. By a simple manipulation and Lemma~\ref{lem:rob}, for sufficiently small $\epsilon$, 
\begin{align}
\label{eq:ex_L1_depend}
\frac{(1-a)a}{2(1-ap)^2} \epsilon \leq \Delta \leq \frac{4}{p} \epsilon.
\end{align}
In particular, if we let $a=2/3$ and $p=1/3$, then $9\epsilon/49 \leq \Delta \leq 12\epsilon$. 
This example shows that the first-order dependence on the \mbox{$\calL_1$-norm} in the bounds of Lemma~\ref{lem:rob} cannot be improved in general. In what follows, this $\calL_1$-norm will be translated into the deviation between the underlying and empirical distributions which vanishes with order $1/\sqrt{n}$ where $n$ is the number of samples.
\end{example}

The next theorem follows from Lemma~\ref{lem:rob} and large deviation results \cite{weissman2003inequalities}. It answers the question raised at the beginning of this section and provides the bounds for the difference between the training and testing privacy-utility guarantees. Note that the bounds provided in the theorem hold for any channel $W_{Y|\hat{X}}$, not just the ones that optimize the PUT. In other words, the theorem holds for any privacy-assuring mapping returned by the  solver in Fig. \ref{fig:robust}, even if this mapping is not a globally optimal solution.

\begin{thm}
\label{thm:robust}
Let $P_{\hat{S},\hat{X}}$ be the empirical distribution obtained from $n$ i.i.d. samples drawn from the true distribution $P_{S,X}$. In addition, denote by $Y$ and $\hat{Y}$ the random variables obtained by passing $X$ and $\hat{X}$ through a given channel $W_{Y|\hat{X}}$, respectively. Let
\begingroup
\allowdisplaybreaks
\begin{align}
    m_S&\defined \left(\min\{P_{S}(s)\mid s\in \mathcal{S}\} - \sqrt{\frac{2}{n}\left(M-\ln\beta\right)}\right)_+,\\
    m_X&\defined \left(\min\{P_{X}(x)\mid x\in \mathcal{X}\} - \sqrt{\frac{2}{n}\left(M-\ln\beta\right)}\right)_+.
\end{align}
\endgroup
Then, with probability at least $1-\beta$,
\begin{align}
    |\chi^2(S;Y)-\chi^2(\hat{S};\hat{Y})|
    \leq & \frac{4}{m_S}\sqrt{\frac{2}{n}\left(M-\ln\beta\right)},\\
    |\chi^2(X;Y)-\chi^2(\hat{X};\hat{Y})|
    \leq & \frac{4}{m_X}\sqrt{\frac{2}{n}\left(M-\ln\beta\right)},
\end{align}
where $M=|\mathcal{S}||\mathcal{X}|$.
\end{thm}
\begin{proof}
See Appendix~\ref{subsec::thm_robust}.
\end{proof}

Note that the bounds provided in Theorem~\ref{thm:robust} depend on the probability of the least likely symbols. The bounds become weaker as $m_X$ and $m_S$ become smaller. We refer the reader to Wang~\etal \cite{wang2018utility} for an alternate approach if $m_X$ and $m_S$ are near zero.
\section{Numerical Results}
\label{sec:numericalresults}
We illustrate some of the results derived in this paper through two experiments. The first experiment, conducted on a synthetic dataset, verifies the tightness of the upper bound for the $\chi^2$-privacy-utility function. The second experiment, run on a real-world dataset, demonstrates the performance of the optimization methods proposed in Section~\ref{sec:add_opt}.
\subsection{Example 1: Parity Bits}
\label{sec:ParityBits}

\begin{figure}
  \centering
  \includegraphics[width=0.52\textwidth]{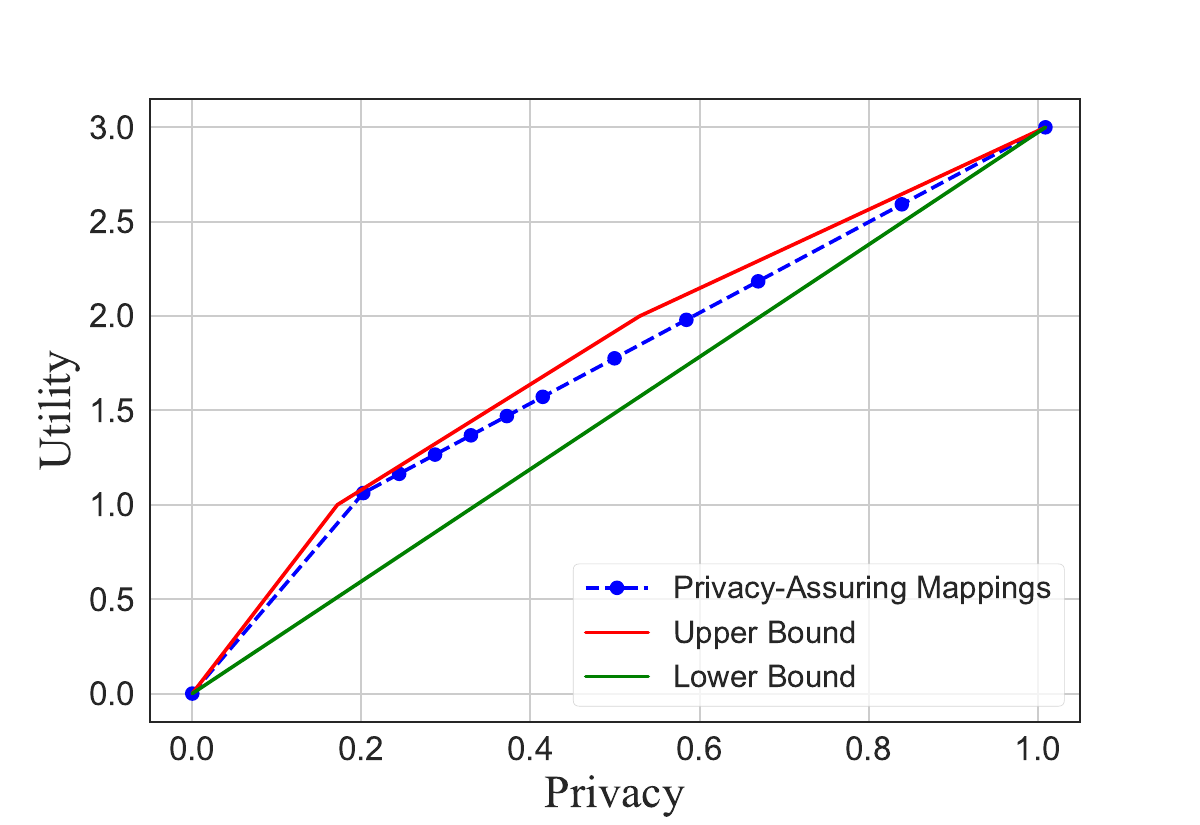}
  \caption{\small{We depict the bounds of the $\chi^2$-privacy-utility function (see Theorem~\ref{thm:bound}) and the privacy-utility values of the privacy-assuring mappings designed by the optimization methods in Section~\ref{sec:add_opt}.
  }}
  \label{fig:MMSE_BitsResult}
\end{figure}
We choose private variable $S=(S_1,S_2) \in \{-1,1\}^2$, where $S$ is composed by two independent bits with $\Pr(S_1=1)=0.45$ and $\Pr(S_2=1)=0.4$. The useful variable $X=(X_1,X_2)\in \{-1,1\}^2$ is generated by passing $S_1$ and $S_2$ through $\mathsf{BSC}(0.2)$ and $\mathsf{BSC}(0.15)$, respectively.

We use the optimization methods proposed in Section~\ref{sec:add_opt} to design privacy-assuring mappings. The private and useful functions are selected as $s_1(S)=S_1$ and $u_1(X)=X_1X_2$, $u_2(X)=X_2$, respectively. We first project the private function to the useful variable. Then we apply Formulation~\ref{alg:PIC-based} with $\mathsf{obj}(\sigma_1,...,\sigma_{n'})=\sum_{i=1}^{n'} \sigma_i$ to find the privacy-assuring mappings. 

In Fig.~\ref{fig:MMSE_BitsResult}, we depict the privacy and utility, measured by $\chi^2$-information, of the privacy-assuring mappings. We also draw the upper bound and lower bound of the $\chi^2$-privacy-utility function. As shown, the privacy-utility values of the designed mappings are very close to the upper bound.
In particular, since the $\chi^2$-privacy-utility function is a concave function (see Lemma~\ref{lem:concave}), the curve of this function is between its upper bound (red line) and the linear interpolation of the achievable privacy-utility values (dashed line).

\subsection{Example 2: UCI Adult Dataset}
\label{sec:Adult}

\begin{figure}[!tb]
  \centering
  \includegraphics[width=0.48\textwidth]{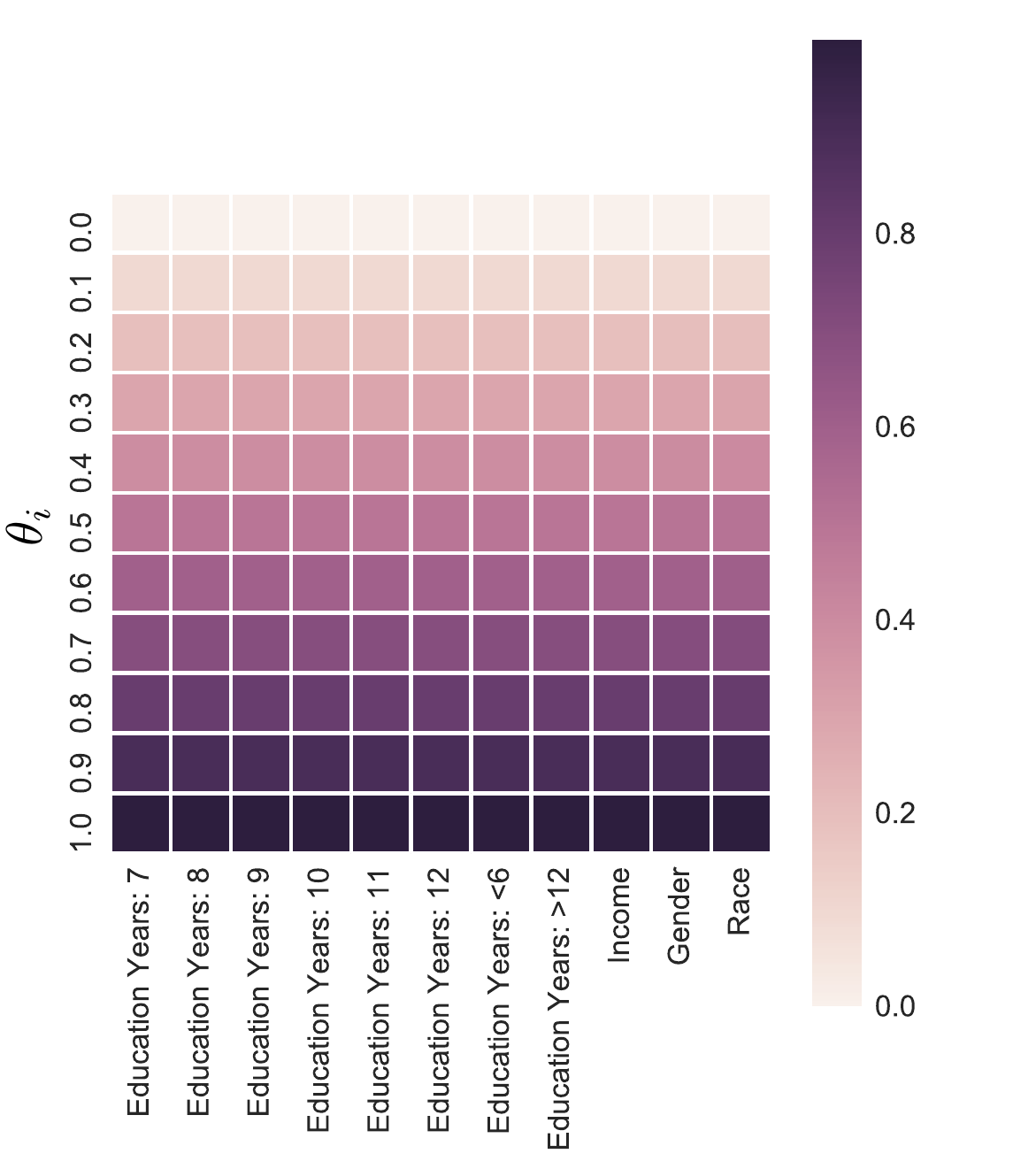}
  \includegraphics[width=0.48\textwidth]{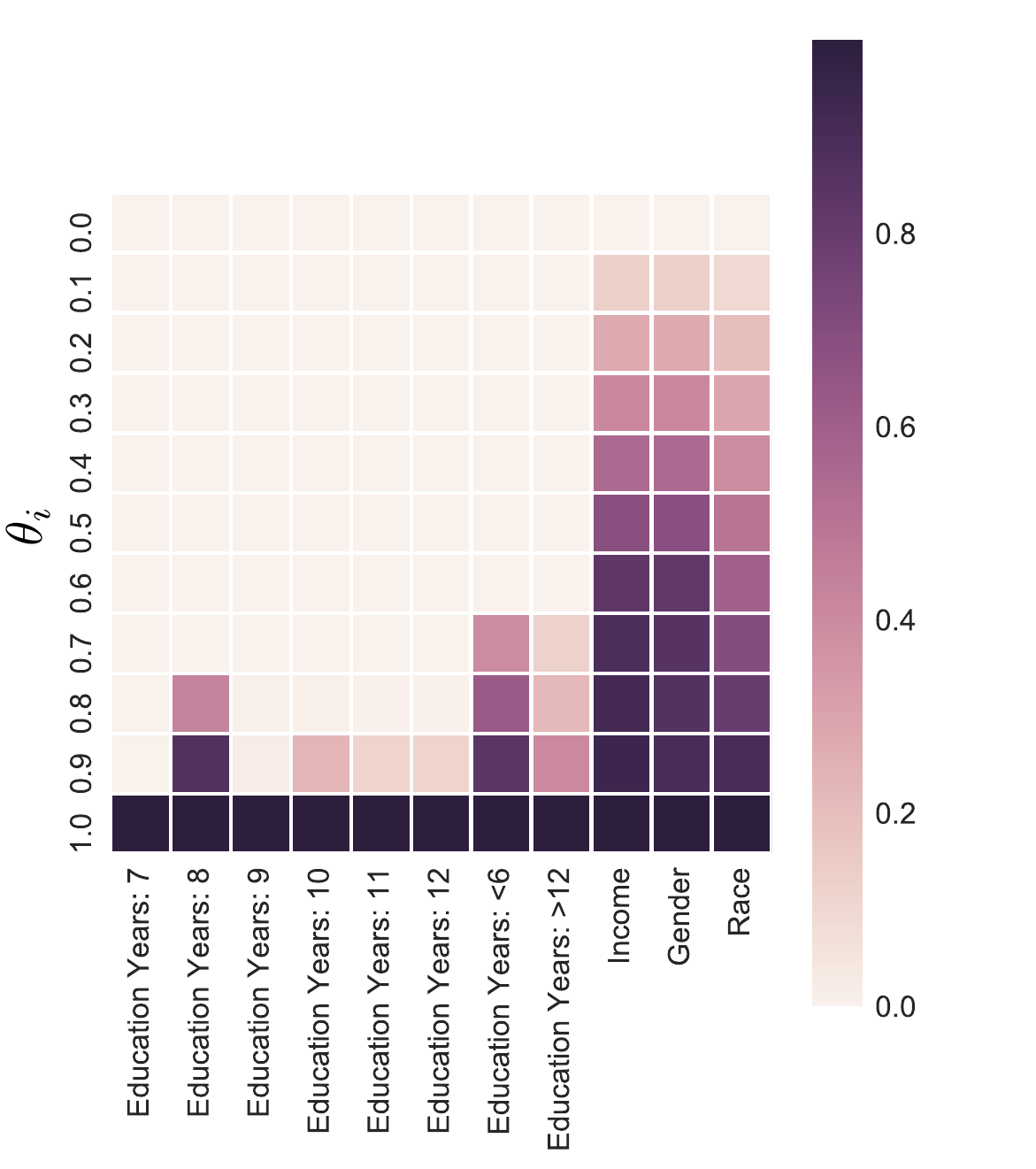}
  \caption{\small{MMSE of estimating each function given the disclosed variable, where darker means harder to estimate. Here $(\textfn{Education Years},\textfn{Income})$ and $(\textfn{Gender},\textfn{Race})$ are useful variable and private variable, respectively. The privacy parameters $\theta_i$ are selected as the same for all $i$ and increase from $0$ to $1$ (i.e., the privacy constraints are increasing from the top down). The privacy-assuring mappings are designed by Formulation~\ref{alg:PIC-based} with $\mathsf{obj}(\sigma_1,...,\sigma_{n'})=\min\{\sigma_1,...,\sigma_{n'}\}$ (left) and with $\mathsf{obj}(\sigma_1,...,\sigma_{n'})=\sum_{i=1}^{n'} \sigma_i$ (right), respectively.}}\label{fig:AdultResult}
\end{figure}

We apply our formulations to the UCI Adult Dataset \cite{Lichman:2013}. A natural selection for the private and useful variables are $S=(\textfn{Gender}, \textfn{Race})$ and $X=(\textfn{Education Years}, \textfn{Income})$, respectively. This allows us to interpret the results of our formulations in an intuitive way, as one would expect there to exist correlations between the chosen private and useful variables. Private functions and useful functions are represented by indicator functions. Furthermore, functions which are linear combinations of others are removed. Following the same procedure proposed in Section~\ref{sec:add_opt}, we first project all private functions to the useful variable. We use QR decomposition \cite{van1996matrix} to construct the basis $\{f_k(x)\}$. Note that other decomposition methods can also be used for constructing basis and, in fact, different bases affect the behavior of the PIC-based convex program (e.g., the joint distribution matrix $\mathbf{P}_{X,Y}$ returned by the optimization may be different). Consequently, the solution produced from the optimization program may not be optimal. Finally, Formulation~\ref{alg:PIC-based} is used to compute the privacy-assuring mappings.

In Fig.~\ref{fig:AdultResult}, we show the MMSE of estimating useful functions and private functions given the disclosed variable. As shown, when we use Formulation~\ref{alg:PIC-based} with $\mathsf{obj}(\sigma_1,...,\sigma_{n'})=\min\{\sigma_1,...,\sigma_{n'}\}$ to compute privacy-assuring mappings, the estimation errors behave uniformly among all functions. This is because we aim at maximizing the worst-case utility over all useful functions. On the other hand, the privacy-assuring mappings designed by Formulation~\ref{alg:PIC-based} with $\mathsf{obj}(\sigma_1,...,\sigma_{n'})=\sum_{i=1}^{n'} \sigma_i$ reveal more interpretable relationships between the private functions and useful functions. We see that $\textfn{Income}$, $\textfn{Gender}$, and $\textfn{Race}$ are highly correlated, and it is not possible to reveal $\textfn{Income}$ while maintaining privacy for $\textfn{Gender}$ and $\textfn{Race}$. Of particular interest are the subtle correlations between the three aforementioned functions and $\textfn{Education Years}$. There is a marked correlation between $\textfn{Education Years}<6$ and, to a lesser degree, $\textfn{Education Years}>12$, with $\textfn{Gender}$ and $\textfn{Race}$. This may be due to the fact that most members of the dataset do not end their education midway. That is, most individuals will either never have begun schooling in the first place or will not continue their education after the 12-year benchmark, which marks graduation from high school. Therefore, we observe that the relationship between $\textfn{Education Years}$ and $\textfn{Race}$ is manifested the most in the two extremities of $\textfn{Education Years}$ ($>12$ and $<6$). Also of note is the correlation between the private functions and $\textfn{Education Years}:8$. Though not as obvious, this relationship can, too, be explained by the fact that 8 years of education marks another benchmark: the beginning of high school, also a time when people are prone to terminating their education.

\section{Conclusion}
\label{sec:conclusion}
In this paper, we studied a fundamental PUT in data disclosure, where an analyst is allowed to reconstruct certain functions of the data, while other private functions should not be estimated with distortion below a certain threshold. First, $\chi^2$-information was used to measure both privacy and utility. Bounds on the best PUT were provided and the upper bound, in particular, was shown to be achievable in the high-privacy region. Moreover, a PIC-based convex program was proposed to design privacy-assuring mappings when the useful functions and private functions were known beforehand. We also derived lower bounds on the MMSE of estimating a target function from the disclosed data and analyzed the robustness of our method when the designer used empirical distribution to compute the privacy-assuring mappings. Finally, we performed two experiments and analyzed the numerical results. Our hope is that the methods presented here can inspire new, information-theoretically grounded and interpretable privacy mechanisms.

\section*{Acknowledgment}

The authors would like to thank the anonymous reviewers and the Associate Editor for their careful reading of our manuscript and their many insightful comments and suggestions.

\appendices
\section{Proofs from Section~\ref{sec:u-p tradeoff}}
\label{Append:A}

\subsection{Lemma~\ref{lem:trace}}
\label{subsec::Lem_trace}
\begin{proof}
By the definition of $\chi^2$-information,
\begin{equation*}
\begin{aligned}
\chi^2(X;Y)+1  = \sum_{x=1}^{|\mathcal{X}|}\sum_{y=1}^{|\mathcal{Y}|}\frac{P_{X,Y}(x,y)}{P_X(x)P_Y(y)}P_{X,Y}(x,y).
\end{aligned}
\end{equation*}
Note that $\bQ_{X,Y}=\bD_X^{-\frac{1}{2}} \bP_{X,Y} \bD_Y^{-\frac{1}{2}}$ which implies
\begin{equation*}
\begin{aligned}
\tr(\bQ_{X,Y}\bQ^T_{X,Y}) 
& = \tr(\bD_X^{-\frac{1}{2}} \bP_{X,Y} \bD_Y^{-1} \bP^T_{X,Y} \bD_X^{-\frac{1}{2}})
= \tr(\bD_X^{-1} \bP_{X,Y} \bD_Y^{-1} \bP^T_{X,Y})\\
&= \sum_{x=1}^{|\mathcal{X}|}\sum_{y=1}^{|\mathcal{Y}|}\frac{P_{X,Y}(x,y)}{P_X(x)}\frac{P_{X,Y}(x,y)}{P_Y(y)}.
\end{aligned}
\end{equation*}
Therefore,
$$\chi^2(X;Y) = \tr(\bQ_{X,Y}\bQ_{X,Y}^T) -1 = \tr(\bA)-1.$$
Since
\begin{equation*}
\begin{aligned}
\bQ_{S,Y} 
= \bD_S^{- \frac{1}{2}}\bP_{S,Y}\bD_Y^{- \frac{1}{2}}
= \bD_S^{- \frac{1}{2}}\bP_{S,X}\bD_X^{- \frac{1}{2}}\bD_X^{-\frac{1}{2}}\bP_{X,Y}\bD_Y^{-\frac{1}{2}}
= \bQ_{S,X}\bQ_{X,Y},
\end{aligned}
\end{equation*}
then
\begin{equation*}
\begin{aligned}
\chi^2(S;Y) 
= \tr(\bQ_{S,Y}\bQ_{S,Y}^T) - 1
= \tr(\bQ_{S,X}\bQ_{X,Y}\bQ_{X,Y}^T\bQ_{S,X}^T) - 1
= \tr(\bB\bA) -1.
\end{aligned}
\end{equation*}
\end{proof}

\subsection{Lemma~\ref{lem:concave}}
\label{subsec::Lem_concave}
\begin{proof}
For $0\leq \epsilon_1 < \epsilon_2 < \epsilon_3 \leq \chi^2(S;X)$,
it suffices to show that
\begin{align*}
    \frac{F_{\chi^2}(\epsilon_3;P_{S,X})-F_{\chi^2}(\epsilon_1;P_{S,X})}{\epsilon_3-\epsilon_1} 
    \leq \frac{F_{\chi^2}(\epsilon_2;P_{S,X})-F_{\chi^2}(\epsilon_1;P_{S,X})}{\epsilon_2-\epsilon_1},
\end{align*}
which is equivalent to
\begin{align}
\frac{\epsilon_2-\epsilon_1}{\epsilon_3-\epsilon_1}F_{\chi^2}(\epsilon_3;P_{S,X})+\frac{\epsilon_3-\epsilon_2}{\epsilon_3-\epsilon_1}F_{\chi^2}(\epsilon_1;P_{S,X})
\leq F_{\chi^2}(\epsilon_2;P_{S,X}).\label{eq:concave}
\end{align}
Let $P_{Y_1|X}$ and $P_{Y_3|X}$ be two optimal solutions in $\mathcal{D}(\epsilon_1;P_{S,X})$ and $\mathcal{D}(\epsilon_3;P_{S,X})$, respectively. Assume that $Y_1$ and $Y_3$ take values in $[m_1]$ and $[m_3]$, respectively. Furthermore, we denote $\lambda \defined \frac{\epsilon_2 - \epsilon_1}{\epsilon_3-\epsilon_1}$. Next, we introduce a new privacy-assuring mapping defined as
\begin{equation}
\begin{aligned}
    P_{Y_\lambda|X}(y|x)
    \defined 
    \begin{cases}
    \lambda P_{Y_3|X}(y|x) & \text{if } y\in [m_3],\\
    (1-\lambda) P_{Y_1|X}(y-m_3|x) & \text{if } y-m_3\in [m_1].
    \end{cases}
\end{aligned}
\end{equation}
Consequently, we have
\begin{equation*}
\begin{aligned}
    P_{Y_\lambda}(y) 
    = \begin{cases}
    \lambda P_{Y_3}(y) & \text{if } y\in [m_3],\\
    (1-\lambda) P_{Y_1}(y-m_3) & \text{if } y-m_3\in [m_1].
    \end{cases}
\end{aligned}
\end{equation*}
Then
\begin{equation*}
\begin{aligned}
\chi^2(X;Y_\lambda)
&=\EE{\frac{P_{X,Y_\lambda}(X,Y_\lambda)}{P_X(X)P_{Y_\lambda}(Y_\lambda)}}-1\\
&=\sum_{y\in [m_3]}\sum_{x=1}^{|\mathcal{X}|}\frac{P_{X,Y_\lambda}(x,y)^2}{P_X(x)P_{Y_\lambda}(y)}+\sum_{y-m_3\in [m_1]}\sum_{x=1}^{|\mathcal{X}|}\frac{P_{X,Y_\lambda}(x,y)^2}{P_X(x)P_{Y_\lambda}(y)}-1\\
&=\sum_{y\in [m_3]}\sum_{x=1}^{|\mathcal{X}|}\frac{P_{Y_\lambda|X}(y|x)^2P_X(x)}{P_{Y_\lambda}(y)}+\sum_{y-m_3\in [m_1]}\sum_{x=1}^{|\mathcal{X}|}\frac{P_{Y_\lambda|X}(y|x)^2P_X(x)}{P_{Y_\lambda}(y)}-1\\
&=\sum_{y\in [m_3]}\sum_{x=1}^{|\mathcal{X}|}\frac{\lambda^2 P_{Y_3|X}(y|x)^2P_X(x)}{\lambda P_{Y_3}(y)}+\sum_{y\in [m_1]}\sum_{x=1}^{|\mathcal{X}|}\frac{(1-\lambda)^2P_{Y_1|X}(y|x)^2P_X(x)}{(1-\lambda)P_{Y_1}(y)}-1\\
&=\lambda \chi^2(X;Y_3) + (1-\lambda)\chi^2(X;Y_1).
\end{aligned}
\end{equation*}
Similarly, we have
\begin{align}
\chi^2(S;Y_\lambda)=\lambda \chi^2(S;Y_3) + (1-\lambda)\chi^2(S;Y_1)
\leq \epsilon_2,
\end{align}
which implies that $P_{Y_\lambda|X} \in \mathcal{D}(\epsilon_2;P_{S,X})$. Therefore,
\begin{align}
F_{\chi^2}(\epsilon_2;P_{S,X})
&\geq \chi^2(X;Y_\lambda)\nonumber\\
&=\lambda \chi^2(X;Y_3) + (1-\lambda)\chi^2(X;Y_1)\nonumber\\
&=\frac{\epsilon_2-\epsilon_1}{\epsilon_3-\epsilon_1}F_{\chi^2}(\epsilon_3;P_{S,X})+\frac{\epsilon_3-\epsilon_2}{\epsilon_3-\epsilon_1}F_{\chi^2}(\epsilon_1;P_{S,X}),
\end{align}
which implies that \eqref{eq:concave} is true, so $F_{\chi^2}(\epsilon;P_{S,X})$ is a concave function. Furthermore, 
$\epsilon \to \frac{1}{\epsilon}F_{\chi^2}(\epsilon;P_{S,X})$ is non-increasing since $F_{\chi^2}(\epsilon;P_{S,X})$ is non-negative and concave.
\end{proof}

\subsection{Closed-Form Expression of $G^m_{\epsilon}(t_1,...,t_n)$}
\label{subsec::closed_form_G}
Recall that, for $t_i \in [0,1]\ (i\in [n])$, $0\leq \epsilon\leq \sum_{i\in [n]} t_i$, and $n\leq m$, $G^m_{\epsilon}(t_1,...,t_n)$ is defined as:
\begin{equation*}
\begin{aligned}
G^m_{\epsilon}(t_1,...,t_n)
\defined \max \left\{ \sum_{i=1}^m x_i\ \Big|\ (x_1,...,x_m) \in \mathcal{D}^m_{\epsilon}(t_1,...,t_n)\right\},
\end{aligned}
\end{equation*}
where 
\begin{equation*}
\begin{aligned}
\mathcal{D}^m_{\epsilon} (t_1,...,t_n)
\defined \left\{(x_1,...,x_m)\ \Big|\ \sum_{i=1}^n t_i x_i  \leq \epsilon, x_i \in [0,1], i\in[m]\right\}.
\end{aligned}
\end{equation*}
We assume $1\geq t_1 \geq ... \geq t_{n-s} > t_{n-s+1}=...=t_n =0$ without loss of generality. Then we divide $\left[0,\sum_{i=1}^{n} t_i\right]$ into $n-s$ intervals:
\begin{align*}
    \left[0,\sum_{i=1}^{n} t_i\right] = \bigcup_{j=0}^{n-1-s}\left[\sum_{i=n-s-j+1}^{n-s}t_i, \sum_{i=n-s-j}^{n-s}t_i\right].
\end{align*}
If $\epsilon \in \left[\sum_{i=n-s-j+1}^{n-s}t_i, \sum_{i=n-s-j}^{n-s}t_i\right]$,
then
\begin{equation}
\label{equa:closed_form_func_G}
G^m_{\epsilon}(t_1,...,t_n) = s+(m-n)+j+\frac{\epsilon-\sum_{i=n-s-j+1}^{n-s}t_i}{t_{n-s-j}},
\end{equation}
and it can be achieved by setting
\begin{align*}
&x_i = 1, \text{ for } i=n-s-j+1,...,m,\\
&x_{n-s-j} = \frac{\epsilon-\sum_{i=n-s-j+1}^{n-s}t_i}{t_{n-s-j}},\\
&x_i = 0, \text{ for } i=1,...,n-s-j-1.
\end{align*}

\subsection{Theorem~\ref{thm:bound}}
\label{subsec::thm_bound_put}
\begin{proof}
The lower bound for $F_{\chi^2}(\epsilon;P_{S,X})$ follows immediately from the concavity of $F_{\chi^2}(\epsilon;P_{S,X})$ and 
\begin{align*}
    F_{\chi^2}(0;P_{S,X}) & \geq 0,\\
    F_{\chi^2}\left(\chi^2(S;X);P_{S,X}\right) & = \chi^2(X;X) = |\mathcal{X}|-1.
\end{align*}
Using Lemma~\ref{lem:trace}, the $\chi^2$-privacy-utility function can be simplified as
$$F_{\chi^2}(\epsilon;P_{S,X}) = \max_{P_{Y|X}\in \mathcal{D}(\epsilon;P_{S,X})} \tr(\bA)-1,$$
$$\mathcal{D}(\epsilon;P_{S,X}) = \{P_{Y|X}\mid S\rightarrow X \rightarrow Y, \tr(\bB\bA)-1 \leq \epsilon\},$$
where
\begin{align*}
    \bA = \bQ_{X,Y}\bQ_{X,Y}^T, \bB = \bQ_{S,X}^T\bQ_{S,X}.
\end{align*}
We denote the singular value decomposition of $\bQ_{S,X}$ and $\bQ_{X,Y}$ by $\bQ_{S,X} = \bW\bSigma_{1}\bU^T$ and $\bQ_{X,Y} = \bV\bSigma_{2}\bM^T$, respectively.  Then $\bB= \bU\bSigma_{1}^T\bSigma_{1}\bU^T = \bU\bSigma_{B}\bU^T$, $\bA= \bV\bSigma_{2}\bSigma_{2}^T\bV^T = \bV\bSigma_{A}\bV^T$ where $\bSigma_B \defined \bSigma_{1}^T\bSigma_{1}$ and $\bSigma_{A} \defined \bSigma_{2}\bSigma_{2}^T$.

Let $\bA_1 = \bU^T \bA \bU = \bL\bSigma_A \bL^T$ where $\bL \defined \bU^T \bV$. Suppose the diagonal elements of $\bA_1$ are $a_1, ... ,a_{|\mathcal{X}|}$. Then, from characterization~3 in Theorem~\ref{thm:PIC_Charac}, we have
\begin{equation}
\label{equa:trace_bound}
\tr(\bB\bA)-1 = a_1-1 +\sum_{i=2}^{d+1} \lambda_{i-1}(S;X)a_i .
\end{equation}
Suppose the $i$-th row of $\bL$ is $\mathbf{l}_i=(l_{i,1},...,l_{i,|\mathcal{X}|})$, the $i$-th column of $\bU$ is $\mathbf{u}_i^T$ and $\mathbf{\Sigma}_A = \diag(\sigma_1,...,\sigma_{|\mathcal{X}|})$. By characterization~3 in Theorem~\ref{thm:PIC_Charac}, $\sigma_1=1$, $\sigma_{j+1}=\lambda_j(X;Y)$ for $j=1,...,d$ and $\sigma_{j+1}=0$ for $j=d+1,...,|\mathcal{X}|-1$. Then, for $\forall i\in [|\mathcal{X}|]$,
\begin{align*}
    0\leq a_i =\sum_{j=1}^{|\mathcal{X}|}\sigma_j l^2_{i,j} \leq \sum_{j=1}^{|\mathcal{X}|} l^2_{ij} = 1.
\end{align*}
Since, following from characterization~3 in Theorem~\ref{thm:PIC_Charac}, the first column of $\bU$ and that of $\bV$ are both $(\sqrt{P_X(1)},...,\sqrt{P_X(|\mathcal{X}|)})^T$, then $\mathbf{l}_1 = \mathbf{u}_1 \bV = (1,0,...,0)$. Therefore, $a_1 =\sigma_1 = 1$. If $P_{Y|X}\in \mathcal{D}(\epsilon;P_{S,X})$, then \eqref{equa:trace_bound} shows
\begin{equation*}
\begin{aligned}
\tr(\bB\bA)-1=\sum_{i=2}^{d+1}\lambda_{i-1}(S;X)a_i \leq \epsilon,
\end{aligned}
\end{equation*}
which implies that
$$(a_2, ..., a_{|\mathcal{X}|})\in \mathcal{D}^{|\mathcal{X}|-1}_{\epsilon}(\lambda_1(S;X),...,\lambda_d(S;X)).$$
Thus,
\begin{equation*}
\begin{aligned}
F_{\chi^2}(\epsilon;P_{S,X})
&\leq \max_{\substack{(a_2,...,a_{|\mathcal{X}|})\\ 
\in \mathcal{D}^{|\mathcal{X}|-1}_{\epsilon}(\lambda_1(S;X),...,\lambda_d(S;X))}} \sum_{i=2}^{|\mathcal{X}|} a_i\\
&= G^{|\mathcal{X}|-1}_{\epsilon}(\lambda_1(S;X),...,\lambda_d(S;X)),
\end{aligned}
\end{equation*}
as required.
\end{proof}

\subsection{Corollary~\ref{cor:increasing}}
\label{subsec::cor_increasing_PUT}
\begin{proof}
First, $F_{\chi^2}(\epsilon;P_{S,X})$ is non-decreasing since, for any $0\leq \epsilon_1<\epsilon_2\leq \chi^2(S;X)$, we have $\mathcal{D}(\epsilon_1;P_{S,X}) \subseteq \mathcal{D}(\epsilon_2;P_{S,X})$. Now suppose there exist $\epsilon_1$ and $\epsilon_2$, such that $F_{\chi^2}(\epsilon_1;P_{S,X}) = F_{\chi^2}(\epsilon_2;P_{S,X})$. We denote $\chi^2(S;X)$ by $\epsilon_0$. Since $F_{\chi^2}(\epsilon;P_{S,X})$ is a concave and non-decreasing function, then for any $\epsilon > \epsilon_1$, $F_{\chi^2}(\epsilon;P_{S,X}) = F_{\chi^2}(\epsilon_1;P_{S,X})$. In particular, $F_{\chi^2}(\epsilon_1;P_{S,X}) = F_{\chi^2}(\epsilon_0;P_{S,X}) = |\mathcal{X}|-1$. This contradicts the upper bound of $\chi^2$-privacy-utility function in Theorem~\ref{thm:bound} since the upper bound implies that $F_{\chi^2}(\epsilon;P_{S,X}) < |\mathcal{X}|-1$ when $\epsilon <\epsilon_0$.
\end{proof}

\subsection{Theorem~\ref{thm:highprivacyregion}}
\label{subsec::thm_high_pri_region}
\begin{proof}
Following from characterization~2 in Theorem~\ref{thm:PIC_Charac}, there exists $f\in \mathcal{L}_2(P_X)$ such that $\|f(X)\|_2=1$, $\EE{f(X)}=0$ and $\|\EE{f(X)|S}\|_2^2=\lambda_{\min}(S;X)$. 

Fix $\mathcal{Y}=\{1,2\}$ and the privacy-assuring mapping is defined as
\begin{equation}
P_{Y|X}(y|x) = \frac{1}{2} + (-1)^y \frac{\sqrt{P_{X\min}}f(x)}{2}.
\end{equation}
Since
$$1=||f(X)||_2^2=\sum_{x=1}^{|\mathcal{X}|}f(x)^2P_X(x)\geq f(x)^2P_X(x),$$
for any $x\in[|\mathcal{X}|]$
$$|f(x)|\leq \frac{1}{\sqrt{P_X(x)}} \leq \frac{1}{\sqrt{P_{X\min}}}.$$
Therefore, $\left|\frac{\sqrt{P_{X\min}}f(x)}{2}\right| \leq \frac{1}{2}$, which implies that $P_{Y|X}(y|x)$ is feasible. Furthermore, $P_Y(y) = \frac{1}{2}$ because of $\EE{f(X)}=0$.

\begin{equation*}
\begin{aligned}
\chi^2(X;Y) = \sum_{x=1}^{|\mathcal{X}|} \sum_{y=1}^{|\mathcal{Y}|} \frac{P_{Y|X}(y|x)^2P_X(x)}{P_Y(y)} -1
=\sum_{x=1}^{|\mathcal{X}|}(P_X(x)+ P_{X\min}f(x)^2P_X(x)) -1
=P_{X\min}.
\end{aligned}
\end{equation*}
Since
\begin{equation*}
\begin{aligned}
P_{Y|S}(y|s) &= \sum_{x=1}^{|\mathcal{X}|} P_{Y|X}(y|x)P_{X|S}(x|s)
=\sum_{x=1}^{|\mathcal{X}|} \left(\frac{1}{2}+(-1)^y\frac{\sqrt{P_{X\min}}f(x)}{2}\right)P_{X|S}(x|s)\\
&=\frac{1}{2} + (-1)^y\frac{\sqrt{P_{X\min}}}{2}\EE{f(X)|S=s},
\end{aligned}
\end{equation*}
then
\begin{equation*}
\begin{aligned}
\chi^2(S;Y)
&= \sum_{s=1}^{|\mathcal{S}|} \sum_{y=1}^{|\mathcal{Y}|} \frac{P_{Y|S}(y|s)^2P_S(s)}{P_Y(y)} -1
=\sum_{s=1}^{|\mathcal{S}|}\left(P_S(s)+P_{X\min}\EE{f(X)|S=s}^2P_S(s)\right) -1\\
&=P_{X\min}\lambda_{\min}(S;X).
\end{aligned}
\end{equation*}
Hence, this $Y$ satisfies $\chi^2(X;Y) = P_{X\min}$ and $\chi^2(S;Y) = P_{X\min}\lambda_{\min}(S;X)$.
\end{proof}

\section{Proofs from Section~\ref{sec:add_opt}}
\label{Append:B}

\subsection{Lemma~\ref{lem:project}}
\label{subsec::thm_projection}
\begin{proof}
Suppose $||f(S)||_2 = 1$ without loss of generality. Observe that
\begin{align*}
\EE{\EE{f(S)|X}}=\EE{f(S)} = 0.
\end{align*}
Since $f(S)\rightarrow X \rightarrow Y$, then $\EE{f(S)|X} = \EE{f(S)|X,Y}$.
Therefore,
\begin{equation*}
\begin{aligned}
\mmse\left(\frac{\EE{f(S)|X}}{||\EE{f(S)|X}||_2}\Bigg|Y\right)
&=\frac{\EE{\EE{f(S)|X}^2}-\EE{\EE{\EE{f(S)|X}|Y}^2}}{||\EE{f(S)|X}||_2^2}\\
&=\frac{\EE{\EE{f(S)|X}^2}-\EE{\EE{\EE{f(S)|X,Y}|Y}^2}}{||\EE{f(S)|X}||_2^2}\\
&=1-\frac{\EE{\EE{f(S)|Y}^2}}{||\EE{f(S)|X}||_2^2}\\
&=\EE{f(S)^2}-\frac{\EE{\EE{f(S)|Y}^2}}{||\EE{f(S)|X}||_2^2}\\
&\leq \EE{f(S)^2} - \EE{\EE{f(S)|Y}^2}\\
&=\mmse(f(S)|Y),
\end{aligned}
\end{equation*}
where the last inequality follows from Jensen's inequality:
\begin{align*}
1=\EE{f(S)^2} = \EE{\EE{f(S)^2|X}}
\geq \EE{\EE{f(S)|X}^2} = ||\EE{f(S)|X}||_2^2.
\end{align*}
\end{proof}

\section{Proofs from Section~\ref{sec:MMSE}}
\label{Append:C}
\subsection{Lemma~\ref{lem:quadBound}}
\label{subsec::Lem_quadBound}
\begin{proof}
For fixed $\ba,\bb \in \Reals^n$ where $a_i>0$ and $b_i\geq0$, let $L_P:\Reals^n\rightarrow \Reals$ and $L_D:\Reals^n\rightarrow \Reals$ be given by
\begin{align*}
L_P(\by) &\defined \ba^T\by,\\
L_D(\bu) & \defined \ba^T\bb + \bu^T\bb +\|\bu\|_2.
\end{align*}
Furthermore, we define $\calA(\ba)\defined \left\{ \bu\in
\Reals^n\mid \bu\geq -\ba \right\}$ and $\calB(\bb)\defined \left\{ \by
\in\Reals^n \mid \|\by\|_2\leq 1,\by\leq\bb \right\}$.

Assume, without loss of generality, that $b_1/a_1\leq
b_2/a_2\leq\dots\leq b_n/a_n$, and let $k^*$ be defined in
\eqref{eq:kstar}. Note that $b_1\leq 1$ and for $k\in [k^*]$
\begin{equation*}
\sum_{i=1}^{k}b_i^2\leq
\frac{a_{k}^2}{\sum_{i=k}^n a_i^2}\left(1-\sum_{i=1}^{k-1}
b_i^2\right)^{+}+\sum_{i=1}^{k-1}b_i^2,
\end{equation*}
so $\sum_{i=1}^{k}b_i^2 \leq 1$. Especially, $\sum_{i=1}^{k^*}b_i^2\leq 1$.
For $c_j\defined \sqrt{\frac{\left(1-\sum_{i=1}^{j}
b_i^2\right)}{\|\ba\|_2^2-\sum_{i=1}^{j} a_i^2}}$, let
$$\by^*=(b_1,\dots,b_{k^*},a_{k^*+1}c_{k^*},\dots,a_nc_{k^*})$$ and
$$\bu^*=(-b_1/c_{k^*},\dots,-b_{k^*}/c_{k^*},-a_{k^*+1},\dots,-a_n).$$ From the definition of $k^*$, $\by^*\in\calB(\bb)$ and $\bu^*\in \calA(\ba)$. Furthermore,
\begin{align}
L_P(\by^*) 
&= \ba^T\by^*\nonumber\\
&= \sum_{i=1}^{k^*} a_ib_i +\sum_{i=k^*+1}^n c_{k^*}a_i^2\nonumber\\
&= \sum_{i=1}^{k^*} a_ib_i + \sqrt{\left(\|\ba\|_2^2-\sum_{i=1}^{k^*} a_i^2 \right)\left(1-\sum_{i=1}^{k^*}b_i^2\right)}, \label{eq:opt_val}
\end{align}
and 
\begin{align*}
L_D(\bu^*)
&= \ba^T\bb + {\bu^*}^T\bb + \|\bu^*\|_2\\
&= \sum_{i=1}^{k^*} \left( a_ib_i-\frac{b_i^2}{c_k^*}\right)
 +c_{k^*}^{-1}\sqrt{\sum_{i=1}^{k^*}b_i^2+c_{k^*}^2\left(
\|\ba\|_2^2-\sum_{i=1}^{k^*}a_i^2 \right)}\\
&= \sum_{i=1}^{k^*} a_ib_i + c_{k^*}^{-1}\left( 1-\sum_{i=1}^{k^*}b_i^2
\right)\\
&= \sum_{i=1}^{k^*} a_ib_i + \sqrt{\left( \|\ba\|_2^2-\sum_{i=1}^{k^*} a_i^2 \right)\left(1-\sum_{i=1}^{k^*}
b_i^2\right)}
= L_P(\by^*).
\end{align*}
Since both the primal and the dual achieve the same value at $\by^*$ and $\bu^*$, respectively, it follows that the value $L_P(\by^*)$ given in \eqref{eq:opt_val} is optimal.
\end{proof}

\subsection{Theorem~\ref{thm:loose}}
\label{subsec::thm_loose}
\begin{proof}
Let
\begin{equation}
h(x)\defined 
\begin{cases}
\rho_0^{-1}(\phi(x)-\sum_{i=1}^m \rho_i \phi_i(x)) & \mbox{if } \rho_0>0,\\ 
0 &  \mbox{otherwise.}
\end{cases}
\end{equation}
Note that when $\rho_0>0$, we have $\|h(X)\|_2=1$. Then for any $\psi\in \calL_2(P_Y)$ and $\|\psi(Y)\|_2=1$,
\begin{align*}
\left| \EE{\phi(X)\psi(Y)} \right|
&= \left| \rho_0\EE{h(X)\psi(Y)}+\sum_{i=1}^m \rho_i
\EE{\phi_i(X)\psi(Y)} \right| \\
&\leq  \rho_0 \left| \EE{h(X)\psi(Y)}\right|+\sum_{i=1}^m  \left|\rho_i
\EE{\phi_i(X)\psi(Y)} \right| \\
&= \rho_0 \left| \EE{h(X)(T_{Y|X}\psi)(X)}\right| + \sum_{i=1}^m  \left|\rho_i\EE{\phi_i(X)(T_{Y|X}\psi)(X)} \right|,
\end{align*}
where $T_{Y|X}$ is defined in Section~\ref{sec:notation}. Denoting $|\EE{h(X)(T_{Y|X}\psi)(X)}|\defined x_0$, $|\EE{\phi_i(X)(T_{Y|X}\psi)(X)}|\defined x_i$, $\bx \defined (x_0,x_1,\dots,x_m)$, the last inequality can be rewritten as
\begin{align}
\label{eq:bound_x}
\left| \EE{\phi(X)\psi(Y)} \right| &\leq \brho_0^T \bx.
\end{align}    

Observe that $\|\bx\|_2 \leq 1$ and $x_i\leq \nu_i$ for $i\in [m]$, and the right hand side of \eqref{eq:bound_x} can be maximized over all values of $\bx$ that satisfy these constraints. We assume, without loss of generality, that $\rho_0> 0$ (otherwise set $x_0=0$). The left-hand side of \eqref{eq:bound_x} can be further bounded by 
\begin{equation}
\left| \EE{\phi(X)\psi(Y)} \right| \leq L_{m+1}(\brho_0,\bnu_0),
\end{equation}
where $\bnu_0=(1,\nu_1,\dots,\nu_m)$ and $L_{m+1}$ is defined in \eqref{eq:defn_Ln}. The result follows directly from Lemma~\ref{lem:quadBound} and noting that $\max_{\psi\in\calL_2(P_Y)}\EE{\phi(X)\psi(Y)}=\|\EE{\phi(X)|Y}\|_2.$
\end{proof}

\subsection{Theorem~\ref{thm:tighter}}
\label{subsec::thm_tighter}
\begin{proof}
For any $\psi\in \calL_2(P_Y)$ with $\|\psi(Y)\|_2=1$, let $\alpha_i\defined \EE{\psi(Y)\psi_i(Y)}$, $\alpha_0\defined \sqrt{1-\sum_{i=i}^t\alpha_i^2}$ and $\psi_0(Y)\defined \alpha_0^{-1}(\psi(Y)-\sum_{i=1}^t \alpha_i \psi_i(Y))$ if $\alpha_0>0$, otherwise $\psi_0(Y)\defined 0$. Observe that $\|\psi_0(Y)\|_2=1$ when $\alpha_0>0$. Also, $\EE{\phi_i(X)\psi_j(Y)}=0$ for $i\neq j$, $i\in \{0,\dots,m\}$, $j\in[t]$.
Consequently,
\begin{align}
\EE{\phi(X)\psi(Y)}
&=\EE{\left(\sum_{i=0}^m \rho_i\phi_i(X)  \right)\left(
\sum_{j=0}^t \alpha_j \psi_j(Y)\right)} \nonumber\\
&=\sum_{i=0}^m\sum_{j=0}^t \rho_i\alpha_j \EE{\phi_i(X)\psi_j(Y)} \nonumber\\ 
&\leq \left| \alpha_0 \sum_{i=0,i\notin [t] }^m \rho_i\EE{\phi_i(X)\psi_0(Y)}\right|+\sum_{i=1}^t|\nu_i\rho_i\alpha_i| \nonumber\\
&\leq |\alpha_0| B_{m-t}\left(\tilde{\brho},\tilde{\bnu}\right)+\sum_{i=1}^t|\nu_i\rho_i\alpha_i| \label{eq:prooftight1}\\
&\leq \sqrt{\sum_{k=1}^t \nu_i^2\rho_i^2 + B_{m-t}\left(\tilde{\brho},\tilde{\bnu}\right)^2} \label{eq:prooftight2}.
\end{align}
Inequality~\eqref{eq:prooftight1} follows from the similar proof of Theorem~\ref{thm:loose}, and \eqref{eq:prooftight2} follows by observing that $\sum_{i=0}^t\alpha_i^2=1$ and applying Cauchy-Schwarz inequality.

Finally, when $\rho_0=0$ and $t=m$, \eqref{eq:prooftight2} can be achieved with equality by taking 
\begin{align*}
    \psi(Y)=\frac{\sum_{i=1}^m \nu_i\rho_i \psi_i(Y)}{\sqrt{\sum_{i=1}^m\nu_i^2\rho_i^2}}.
\end{align*}
\end{proof}

\section{Proofs from Section~\ref{sec:robust}}
\label{Append:D}

\subsection{Lemma~\ref{lem:rob}}
\label{subsec::Lem_rob}
\begin{proof}
By the definition of $\chi^2$-information, we have
\begin{align}
|\chi^2(X_1;Y_1) - \chi^2(X_2;Y_2)|
&\leq \sum_{x=1}^{|\mathcal{X}|}\sum_{y=1}^{|\mathcal{Y}|}\left|\frac{P_{X_1,Y_1}(x,y)^2}{P_{X_1}(x)P_{Y_1}(y)}-\frac{P_{X_2,Y_2}(x,y)^2}{P_{X_2}(x)P_{Y_2}(y)}\right|. \label{equa:lem_rob_1}
\end{align}
By the triangle inequality,
\begin{align}
\left|\frac{P_{X_1,Y_1}(x,y)^2}{P_{X_1}(x)P_{Y_1}(y)}-\frac{P_{X_2,Y_2}(x,y)^2}{P_{X_2}(x)P_{Y_2}(y)}\right|
&\leq \frac{P_{X_1,Y_1}(x,y)}{P_{X_1}(x)P_{Y_1}(y)}\left|P_{X_1,Y_1}(x,y)-P_{X_2,Y_2}(x,y) \right| \nonumber\\
&\quad+ \frac{P_{X_1,Y_1}(x,y)}{P_{X_1}(x)P_{Y_1}(y)}\frac{P_{X_2,Y_2}(x,y)}{P_{X_2}(x)} \left|P_{X_1}(x)-P_{X_2}(x)\right| \nonumber\\
&\quad+ \frac{P_{X_1,Y_1}(x,y)}{P_{Y_1}(y)}\frac{P_{X_2,Y_2}(x,y)}{P_{X_2}(x)P_{Y_2}(y)} \left|P_{Y_1}(y)-P_{Y_2}(y)\right| \nonumber\\
&\quad+ \frac{P_{X_2,Y_2}(x,y)}{P_{X_2}(x)P_{Y_2}(y)}\left|P_{X_1,Y_1}(x,y)-P_{X_2,Y_2}(x,y)\right|.\label{equa:lem_rob_2}
\end{align}
Note that for $i=1,2$
\begin{equation*}
\frac{P_{X_i,Y_i}(x,y)}{P_{X_i}(x)P_{Y_i}(y)} \leq \frac{1}{P_{X_i}(x)}\leq \frac{1}{m_X}.
\end{equation*}
Therefore, we have
\begin{align}
|\chi^2(X_1;Y_1) - \chi^2(X_2;Y_2)|
&\leq \frac{2}{m_X}\|P_{X_1,Y_1}-P_{X_2,Y_2}\|_1
+\frac{1}{m_X}(\|P_{X_1}-P_{X_2}\|_1+\|P_{Y_1}-P_{Y_2}\|_1)\nonumber\\
&\leq \frac{4}{m_X}\|P_{X_1,Y_1}-P_{X_2,Y_2}\|_1. \label{equa:lem_rob_5}
\end{align}
Similarly,
\begin{align}
|\chi^2(S_1;Y_1) - \chi^2(S_2;Y_2)|
\leq \frac{4}{m_S}\|P_{S_1,Y_1}-P_{S_2,Y_2}\|_1. \label{equa:lem_rob_6}
\end{align}
By the data processing inequality and the assumption $P_{Y_1|X_1}=P_{Y_2|X_2}$, we have 
\begin{align*}
\|P_{S_1,Y_1}-P_{S_2,Y_2}\|_1
&\leq \|P_{S_1,X_1}-P_{S_2,X_2}\|_1,\\
\|P_{X_1,Y_1}-P_{X_2,Y_2}\|_1
&\leq \|P_{S_1,X_1}-P_{S_2,X_2}\|_1.
\end{align*}
Therefore, we get the desired conclusion.
\end{proof}

\subsection{Theorem~\ref{thm:robust}}
\label{subsec::thm_robust}
Recall the following results by Weissman~\etal~\cite[Theorem~2.1]{weissman2003inequalities} for the $\calL_1$ deviation of the empirical distribution. 

For all $\epsilon > 0$,
\begin{equation}
    \Pr\left(\|\hat{P}_n-P\|_1 \geq \epsilon\right) 
    \leq (2^M-2)\exp(-n\bar{\phi}(\pi_P)\epsilon^2/4),
\end{equation}
where $P$ is a probability distribution on the set $[M]$, $\hat{P}_n$ is the empirical distribution obtained from $n$ i.i.d. samples,
$\pi_P \defined \max_{\mathcal{M}\subseteq[M]}\min(P(\mathcal{M}),1-P(\mathcal{M}))$, and 
\begin{equation*}
    \bar{\phi}(p)\defined 
    \begin{cases}
    \frac{1}{1-2p}\ln \frac{1-p}{p} & p\in[0,1/2), \\ 
    2 & p=1/2.
    \end{cases}
\end{equation*}
Note that $\bar{\phi}(\pi_P)\geq 2$, which implies that
\begin{equation}
\label{equa:L1bound_weissman}
\begin{aligned}
    \Pr\left(\|\hat{P}_n-P\|_1 \geq \epsilon\right)
    &\leq \exp(M)\exp(-n\epsilon^2/2).
\end{aligned}
\end{equation}
Therefore, by choosing $P=P_{S,X}$, $M=|\mathcal{S}||\mathcal{X}|$, and $\epsilon = \sqrt{\frac{2}{n}\left(M-\ln\beta\right)}$, \eqref{equa:L1bound_weissman} implies that, with probability at least $1-\beta$,
\begin{equation}
\label{equa:L1bound_thm}
    \|P_{\hat{S},\hat{X}} - P_{S,X}\|_1 \leq \sqrt{\frac{2}{n}\left(M-\ln\beta\right)},
\end{equation}
where $P_{\hat{S},\hat{X}}$ is the empirical distribution obtained from $n$ i.i.d. samples drawn from $P_{S,X}$. Also of note,
\begin{align}
\min\{P_{\hat{S}}(s)\mid s\in\mathcal{S}\}
&\geq \min\{P_{S}(s)\mid s\in\mathcal{S}\} - \|P_{\hat{S}}-P_S\|_1 \nonumber\\
&\geq \min\{P_{S}(s)\mid s\in\mathcal{S}\} - \|P_{\hat{S},\hat{X}}-P_{S,X}\|_1. \label{eq::min_ms}
\end{align}
Similarly, 
\begin{align}
\min\{P_{\hat{X}}(x)\mid x\in\mathcal{X}\} 
\geq \min\{P_{X}(x)\mid x\in\mathcal{X}\} - \|P_{\hat{S},\hat{X}}-P_{S,X}\|_1. \label{eq::min_mx}
\end{align}
The proof of Theorem~\ref{thm:robust} then follows from Lemma~\ref{lem:rob} and \eqref{equa:L1bound_thm}, \eqref{eq::min_ms}, \eqref{eq::min_mx}.

\bibliographystyle{IEEEtran}
\bibliography{references}

\end{document}